\documentclass[11pt,aps,pra,showpacs,superscriptaddress,floatfix]{revtex4}
\usepackage{epsfig}
\usepackage{amsmath, amssymb, amsfonts, mathrsfs}
\usepackage{dsfont}
\usepackage{amsthm}
\usepackage{graphicx}
\usepackage{color}
\usepackage{hyperref}
\usepackage{bbm}
\usepackage{soul}
\usepackage{ulem}
\usepackage{braket}

\newcommand{\de}[1]{\left( #1 \right)}
\newcommand{\De}[1]{\left[#1\right]}
\newcommand{\DE}[1]{\left\{#1\right\}}
\newcommand{\ketbra}[2]{\left|#1\middle\rangle\middle\langle#2\right|}
\newcommand{\Hi}{\mathcal{H}}

\newcommand{\I}{\mathds{1}}
\newcommand{\OI}{I}
\newcommand*{\coloneqq}{\mathrel{\vcenter{\baselineskip0.5ex \lineskiplimit0pt \hbox{\scriptsize.}\hbox{\scriptsize.}}} =}

\DeclareMathOperator{\Tr}{Tr}

\newcommand{\beq}{\begin{equation}}
\newcommand{\eeq}{\end{equation}}

\newtheorem{theorem}{Theorem}
\newtheorem{proposition}[theorem]{Proposition}
\newtheorem{definition}[theorem]{Definition}
\newtheorem{lemma}[theorem]{Lemma}

\newtheorem{corollary}[theorem]{Corollary}

\newcommand{\be}{\begin{equation}}
\newcommand{\ee}{\end{equation}}

\newcommand{\ben}{\begin{eqnarray}}
\newcommand{\een}{\end{eqnarray}}
\usepackage{amsmath}

\newcommand\myequ[1]{\mathrel{\overset{\makebox[0pt]{\mbox{\normalfont\tiny\sffamily #1}}}{=}}}

\begin{document}

\title{Genuine-multipartite entanglement criteria based on positive maps}
\author{Fabien Clivaz} 
\affiliation{Group of Applied Physics, University of Geneva, 1211 Geneva 4, Switzerland}
\author{Marcus Huber}
\affiliation{Institute for Quantum Optics and Quantum Information (IQOQI), Austrian Academy of Sciences, Boltzmanngasse 3, A-1090 Vienna, Austria}
\author{Ludovico Lami}
\affiliation{F\'isica Te\`orica: Informaci\'o i Fen\`omens Qu\`antics, Universitat Aut\`onoma de Barcelona, ES-08193 Bellaterra (Barcelona), Spain}
\author{Gl\'aucia Murta}
\affiliation{QuTech, Delft University of Technology, Lorentzweg 1, 2628 CJ Delft, the Netherlands}
\affiliation{Departamento de F\'isica, Universidade Federal de Minas Gerais, Caixa Postal 702, 30123-970, Belo Horizonte, MG, Brazil}


\begin{abstract}
Positive maps applied to a subsystem of a bipartite quantum state constitute a central tool in characterising entanglement. In the multipartite case, however, the direct application of a positive but not completely positive map cannot distinguish if a state is genuinely multipartite entangled or just entangled across some bipartition. We thus generalise this bipartite concept to the multipartite setting by introducing non-positive maps that are positive on the subset of bi-separable states, but can map to a non-positive element if applied to a genuine multipartite entangled state. We explicitly construct examples of multipartite non-positive maps, obtained from positive maps via a lifting procedure, that in this fashion can reveal genuine multipartite entanglement in a robust way.
\end{abstract}
\maketitle 

\section{Introduction}

Due to the importance of entanglement as a resource for quantum information processing, the task of determining whether a quantum state is 
entangled or not plays a crucial role for the theoretical and practical developments of the field. 
While the usefulness of bipartite entanglement is well established, the development of applications using entanglement in multipartite systems
is still in its early stages, however, it has already been shown to play a fundamental role for universal quantum computation
in the measurement-based quantum computation paradigm \cite{MBQC} and for cryptographic tasks such as secret sharing \cite{secretsharing}. 

Concerning the characterisation of entanglement, one of the biggest challenges to start with is the fact that it is NP-hard in the Hilbert space dimension to decide whether a given quantum state is entangled \cite{Gurvits}. Since the dimension grows exponentially in the number of involved parties, the exact answer to that question will probably remain elusive for many-body systems.

Nonetheless there exists an abundance of necessary separability criteria capable of certifying entanglement in a practically satisfying way. For bipartite entanglement, the two most paradigmatic techniques are positive maps which are not completely positive \cite{HorodeckiCPmaps} and entanglement witnesses \cite{bruss}. These two concepts are intimately related, as every positive map directly leads to a multitude of entanglement witnesses and every witness can be associated to a positive map \cite{HorodeckiCPmaps}.
Both concepts are sufficient for the characterisation of entanglement, in the sense that, for every entangled state, there exists both a positive map and an entanglement witness certifying its entanglement \cite{HorodeckiCPmaps, bruss}.

Turning to multipartite systems, the potentially diverse separability structures add a layer of complexity to the detection of entanglement. As opposed to detecting any bipartite entanglement in multipartite systems \cite{horodecki2}, mixed multipartite states pose another fundamental challenge. Counterintuitively, multipartite systems can exhibit entanglement across every partition, yet still not feature any genuine multipartite entanglement (GME). To certify GME, an abundance of entanglement witnesses were derived (see e.g. reviews \cite{tothguhne,siewerteltschka}) and a characterisation in terms of semi-definite programs (SDP) was developed \cite{guhnetaming,relaxations}. 
There is however no direct correspondence to the concept of positive map based criteria, not even for the undoubtedly most used criterion of
positivity under partial transposition (PPT). Indeed, the naive application of a positive map to a subsystem in correspondence to the 
bipartite case must inevitably fail, as all it can reveal is entanglement across that bipartition, which, as we just mentioned, is never enough to 
infer GME.

In this manuscript we expand the notion of positive map based criteria from the bipartite to the multipartite case, thus filling a gap in the set of available tools. Our main idea is to derive maps which are positive on all biseparable states, whilst not positive on the set of all states. We first give a general description of the approach and then
showcase 
some exemplary maps. The maps we derive are based on convex combinations of positive maps that are used in the bipartite case, thus generalising the PPT criterion and others to the multipartite case.

The manuscript is organised as follows: In Section \ref{sec.preliminaries} we set the stage by recalling some important bipartite and multipartite entanglement definitions and giving the motivation of this work.
In Section \ref{sec.results} we introduce our framework and general construction method based on positive maps. 
In Section \ref{sec.examples} we give some explicit examples of our method. In particular, a remarkably simple partial transposed based criteria is worked out,
as well as shown to be extendable to other maps such as the Reduction and the Breuer-Hall map, and a Choi map construction is shown to robustly detect noisy $n$-qudit GHZ-like states as well as PPT states. Finally in Section \ref{sec.discussion} we discuss possible extensions and applications of our method.

\section{Preliminaries}\label{sec.preliminaries}

To set the stage and notation we first remind the reader of the concepts of separability, entanglement witnesses and positive maps in the bipartite case. A finite dimensional state is called separable iff it can be decomposed into a convex combinations of pure product states, i.e.
\begin{align}\label{eqsep}
\rho_{\text{sep}}\coloneqq\sum_ip_i|\phi^i_A\rangle\langle\phi^i_A|\otimes|\phi^i_B\rangle\langle\phi^i_B|\,,
\end{align}
where $\{p_i\}$ is a probability distribution.

The states which can be written in the form of  Eq. \eqref{eqsep} form a closed convex set and therefore, as a consequence of the Hahn-Banach theorem (see \cite{HorodeckiCPmaps}), the set of separable states 
can be separated from any point in its complement by a hyperplane that can be written as $\text{Tr}(\rho W)=0$ 
for a self-adjoint operator (observable) $W=W^\dagger$.
Now if (by convention) $\text{Tr}(\rho_{\text{sep}}W)\geq0\) for all \(\rho_{\text{sep}}$ and there exits at least one (entangled) state $\rho$ 
such that $\text{Tr}(\rho W)<0$, the operator $W$ is referred to as an \emph{entanglement witness}. The above thus states that any entangled state can be detected by an entanglement witness $W$.

From another perspective a linear positive map $\Lambda$, $\Lambda[\alpha \rho+\beta \sigma]=\alpha \Lambda[\rho]+\beta\Lambda[\sigma]$, $\Lambda[\rho]\geq 0 \; \forall\, \rho\geq 0$, can
be used to detect entanglement if it is not completely positive, since the application of a non-completely positive map to an 
entangled state $\rho_{AB}\geq 0$ can result in a non-positive operator, $\Lambda_A\otimes\mathbb{I}_B[\rho_{AB}]\ngeq 0$.
The fact that the extension of a positive map remains positive on separable states can be easily seen by applying it 
to the state $\rho_{\text{sep}}$ in Eq. \eqref{eqsep}.
The equivalence of the two approaches was established in Ref. \cite{HorodeckiCPmaps}, where the authors had proven that
for every entangled state $\rho_{AB}$ there exists a positive but 
non-completely positive map, $\Lambda_A$, whose extension applied to $\rho_{AB}$ maps it into an operator with negative eigenvalues.

It is straightforward to generate a witness out of a positive map: if $\Lambda_A\otimes\OI_B[\rho_{AB}]\ngeq 0$,
it implies that there exists a pure state $|\psi\rangle$ such
that $\text{Tr}(|\psi\rangle\langle\psi|\Lambda_A\otimes\OI_B[\rho_{AB}])<0$.
Now, by invoking the concept of the dual of a map \(\Lambda^*\), uniquely defined by 
$\text{Tr}(\rho \Lambda[\rho'])=\text{Tr}(\Lambda^*[\rho]\rho'),\; \forall \rho, \rho'$, it can 
immediately be seen that $\Lambda^*_A\otimes\OI_B[|\psi\rangle\langle\psi|]$ is an entanglement witness detecting \(\rho_{AB}\).
{The converse construction goes as follows. Every witness $W$ is a block-positive operator, i.e. it is positive on product vectors: $\langle \alpha \beta | W |\alpha \beta\rangle\geq 0$. This implies the positivity of the map $\Lambda$ whose Choi matrix is $W$, defined in the computational basis (up to a constant) by $\Lambda(|i\rangle\!\langle j|)=\sum_{l,m} \langle li | W | mj \rangle \, |l\rangle\!\langle m|$.
If $W$ detects $\rho$, then $0>\text{Tr} (W\rho) = \text{Tr}\,\de{ (\Lambda\otimes I)(|\varepsilon\rangle\!\langle \varepsilon|)\, \rho} = \langle \varepsilon | (\Lambda^{*}\otimes I)(\rho) |\varepsilon\rangle$, where $|\varepsilon\rangle = \sum_{i=1}^{d} |ii\rangle$ is the (unnormalised) maximally entangled state on the bipartite system $AB$. Thus, we see that the positive map $\Lambda^{*}$ reveals the entanglement of $\rho$. For details, we refer the reader to \cite{HorodeckiCPmaps}.}

To appreciate the challenges in the multipartite case, we first review the different levels of separability a system can exhibit.
The strongest notion of separability is the complete absence of any type of entanglement, i.e. \emph{full separability}.
An $n$-partite finite dimensional quantum state $\rho_{sep} \in \mathcal{P}(\Hi_1\otimes \ldots \otimes \Hi_n)$ is called fully-separable, and will be denoted by \(\rho_{\text{sep}}\), iff it can be decomposed as
\begin{align}\label{eq.sep}
\rho_{sep}= \sum_i p_i \; \rho^i_1\otimes \ldots \otimes \rho^i_{n} ,
\end{align}
where as before \(\{p_i\}\) is a probability distribution.
If a quantum state $\rho$ cannot be decomposed into the form \eqref{eq.sep} there must be some entanglement in the system. This concept of separability can also be revealed in terms of general linear maps \cite{horodecki2}.
However, as we have mentioned before, for multipartite systems, many levels of separability can exist, according to how many subsystems share 
entanglement. We now present the weakest notion of separability, usually referred to as \emph{biseparability}:
Let $\rho_{2-sep} \in \mathcal{P}(\Hi_1\otimes \ldots \otimes \Hi_n)$ be an $n$-partite   quantum state and let $A  \subset \DE{1,\ldots,n}$ denote
a proper subset of the parties. A state $\rho_{2-sep}$ is biseparable iff it can be decomposed as
\begin{align}\label{eq.2sep}
\rho_{2-sep}=\sum_{A} \sum_i p^i_A \; \rho^i_A\otimes \rho^i_{\bar{A}}, \quad p^i_A\geq 0,\;  \sum_A \sum_{i} p^i_A =1,
\end{align}
where $\rho_A$ denotes a quantum state for the subsystem defined by the subset $A$ and $\sum_A$ stands for the sum over all bipartitions $A \lvert \bar{A}$.

An $n$-partite state which cannot be decomposed as \eqref{eq.2sep} is called \emph{genuine $n$-partite entangled}. If the number of parties is clear from the context, 
we call those states genuine multipartite entangled (GME). 
Note that since the biseparable states form a convex closed set, the Hahn-Banach theorem ensures that given any GME-state there exists a GME-witness detecting it.
Just as in the bipartite case, the witnesses of genuine multipartite entanglement (\emph{GME-witnesses})
are defined by hermitian operators $W_{\text{GME}}$ that, for all $\rho_{2-sep}$, fulfil  $\text{Tr}(\rho_{2-sep}W_{\text{GME}})\geq0$ and for which
there exists a multipartite state $\rho$ such that $\text{Tr}(\rho W_{\text{GME}})<0$. 

However, as mentioned in the introduction, positive maps  fail to capture the concept of genuine multipartite entanglement for a simple reason: applying a positive map to
any marginal subsystem
$A$ of a biseparable state does not necessarily result in a positive operator, as 
it only guarantees positivity for the portion of the state which is separable with respect to the partition $A|\bar{A}$.

Let us remark that, for bipartite systems,  the crucial point of entanglement criteria based on  positive maps \(\Lambda\) is the following:
If we consider their extension \(\Phi\coloneqq \Lambda_A \otimes \OI_B\) they become non-positive maps which are nonetheless positive on all separable states, see Fig. \ref{diagram}. 
This is the point of view that allows for a straightforward multipartite generalisation. The objects we are thus looking for are non-positive maps $\Phi_{\text{GME}}$ such that
\begin{align}
\label{equ_GME-map}
\Phi_{\text{GME}}[\rho_{2-sep}]\geq 0,\quad \forall\,\rho_{2-sep}\,.
\end{align}
We call those maps GME-maps. In the remaining of the manuscript we explore the possibility of constructing such maps from the lifting of positive maps (see Fig.\ref{diagram} for an illustration of our concept).

\begin{figure}[t]
\begin{center}
\includegraphics[width=0.3\columnwidth]{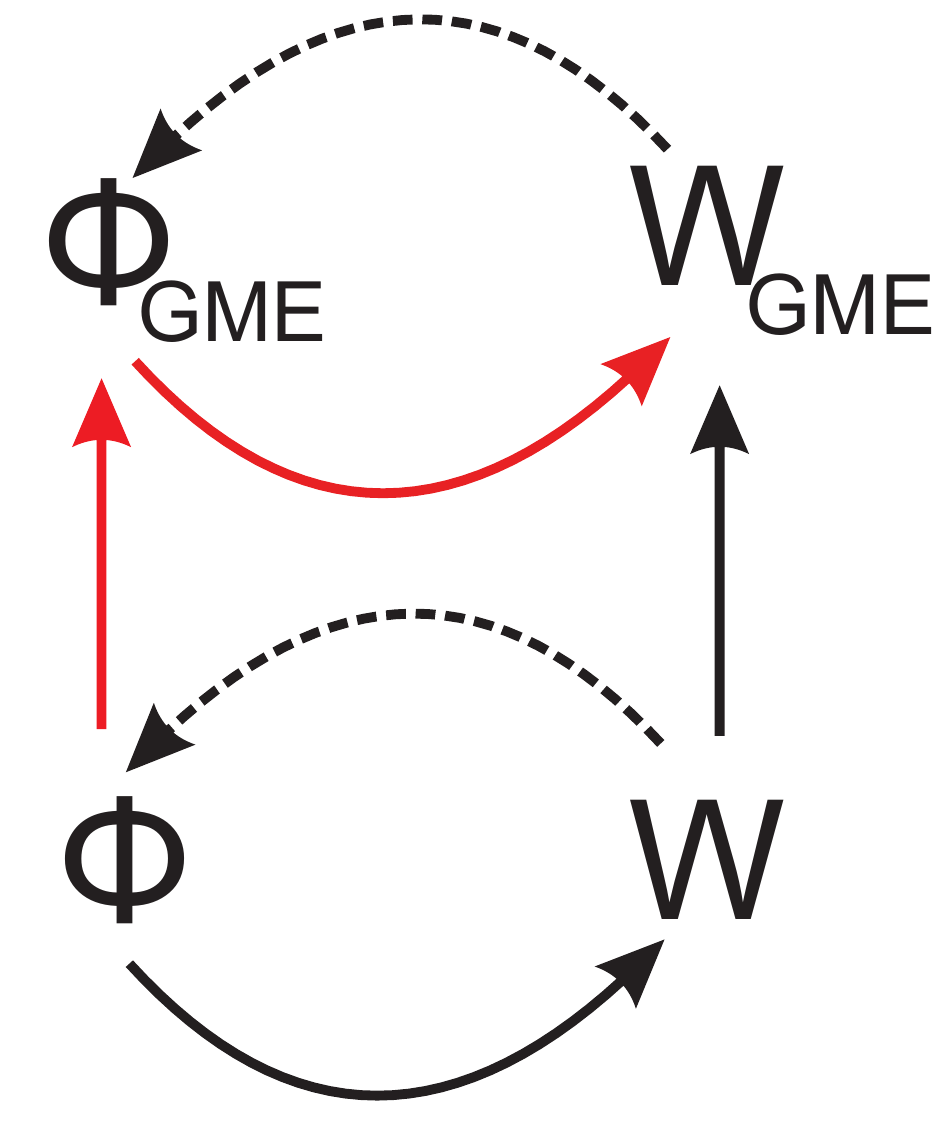}
\end{center}
\caption{An illustration of the connections between entanglement witnesses and non-positive maps. The bottom of the figure illustrates the bipartite situation, where the non-positive maps \(\Phi\) are obtained via tensoring a positive map with the identity. The top represents the multipartite situation, where the non-positive maps \(\Phi_{\text{GME}}\), so called GME-maps, can in some cases be constructed from positive maps via a lifting procedure. \(W\) and \(W_{\text{GME}}\) denote the witnesses of the respective cases. The black arrows denote previously known connections, the red arrows indicate our contribution with this paper.}
\label{diagram} 
\end{figure}

\section{GME criteria based on positive maps}\label{sec.results}

In this Section we present a method of developing GME criteria based on positive but not completely positive maps. Before we introduce our method, we first note that it is always possible to find  a GME-map, as defined in Eq.~\eqref{equ_GME-map}, detecting a given GME-state. Indeed let \(\rho_{\text{GME}}\) be a GME-state. 
 Then we know from \cite{relaxations} that there exists a GME-witness detecting \(\rho_{\text{GME}}\) of the form

\begin{equation}\label{eq.W}
W_{\text{GME}}= \sum_{A} \Lambda_A^* \otimes \OI_{\bar{A}} (\ket{\psi_A} \bra{\psi_A})+M_{\{\Lambda_A, \ket{\psi_A}\}_A},
\end{equation}
where for each partition A, \(\Lambda_A\) is a positive map, \(\ket{\psi_A}\) is chosen such 
that \(\bra{\psi_A} \Lambda_A \otimes \OI_{\bar{A}} (\rho_{\text{GME}} )\ket{\psi_A} < 0\), and \(M_{\{\Lambda_A, \ket{\psi_A}\}_A}\)
is a positive matrix depending on the choice of \(\Lambda_A\) and \(\ket{\psi_A}\). In Appendix \ref{app:W_GME}, it is shown how one obtains \(W_{\text{GME}}\) in \eqref{eq.W} from \cite{relaxations} and in particular how \(M_{\{\Lambda_A, \ket{\psi_A}\}_A}\) is constructed. Then the map
\begin{equation}\label{eq.PhifromW}
\Phi_{\text{GME}}[\rho]=\text{Tr}(W_{\text{GME}} \cdot \rho) \; \I
\end{equation}
 is a GME-map detecting \(\rho_{\text{GME}}\). We should also note that given any GME-map one can associate a GME-witness to each state detected by the map in the same way as done in the bipartite setting. Indeed given \(\Phi_{\text{GME}}\) and \(\rho_{\text{GME}}\) such that \(\Phi_{\text{GME}}(\rho_{\text{GME}}) \ngeq 0\), there exists a pure state \(\ket{\psi}\) such that \(\Phi_{\text{GME}}^* [\ket{\psi} \bra{\psi}]\) is a GME-witness detecting \(\rho_{\text{GME}}\). 
 
This looks all good but upon inspecting the GME-map defined by Eq.\eqref{eq.PhifromW}, one remarks that it is only gained from positive maps in an indirect way. 
 Indeed from the family of maps \(\{\Lambda_A \otimes \OI_{\bar{A}}\}_A\) (bottom left of Fig. \ref{diagram}), a family of bipartite 
 witnesses \(\{\Lambda_A^*\otimes \OI_{\bar{A}}(\ket{\psi_A} \bra{\psi_A})\}_A\) was first associated (bottom right of Fig. \ref{diagram}), 
 from which a GME-witness was extracted (top right of Fig. \ref{diagram}). This GME-witness finally defined the non-positive map \(\text{Tr}(W_{\text{GME}} \cdot \rho) \I\) of Eq.\eqref{eq.PhifromW} (top left of Fig. \ref{diagram}).
 
Now, our goal is to explore a direct lifting method, illustrated by the red arrow from \(\Phi\) to \(\Phi_{\text{GME}}\) in Fig. \ref{diagram}, from positive maps to GME maps, without passing through the witnesses. The reason for pursuing this is that GME-maps of the form of Eq.\eqref{eq.PhifromW} are quite trivial in the sense that the boundary of the set \(\{\rho \mid \text{Tr} (W_{\text{GME}} \cdot \rho) \; \I \geq 0\}\)
is a hyperplane rather than the boundary of a more complex convex set, which could possibly explore more subtleties
of genuine multipartite entanglement.

Inspiring ourselves from the structure of \(W_{\text{GME}}\) in Eq. \eqref{eq.W}, our goal is to seek for maps of the form
{
\begin{equation}
\label{eq:phidef}
\Phi_{\text{GME}}\coloneqq\sum_{A} \Lambda_A \otimes \OI_{\bar{A}} \circ \mathcal{U}^{(A)}  + M,
\end{equation}
where $M$ is a positive map, $\mathcal{U}^{(A)}[\rho]\coloneqq\sum_i p_i^{(A)} \, U_i^{(A)}\rho  \left( U_i^{(A)}\right)^\dagger$ is a family of convex combinations of local unitaries, and $\Phi_{\text{GME}}[\rho_{2-sep}]\geq 0$ $\forall \rho_{2-sep}$.}

 After those general considerations, we want to give in the remaining of the manuscript concrete examples of our method by exhibiting GME-maps constructed via the lift of positive but 
non-completely positive maps.

\section{Some direct lifting examples\label{sec.examples}}
\subsection{\label{sec:trans}Transposition-based GME criteria}

Our first example is possibly the most naive attempt, where we set
$M=c\,\I\cdot \text{Tr}$, \(\mathcal{U}^{(A)} =\OI\) and $\Lambda=T$, the transposition, in a tripartite setting, and find the value of $c$ for which \(\Phi_{\text{GME}}\) defined by Eq. \eqref{eq:phidef} is a GME-map. That is, 
we consider the map
\begin{align}\label{eqmapT}
\Phi_T[\cdot]=(T_A \otimes \OI_B \otimes \OI_C + \OI_A \otimes T_B \otimes \OI_C  + \OI_A \otimes \OI_B \otimes T_C  +c \,\I\cdot \text{Tr})[\cdot].
\end{align}

\begin{theorem}\label{thmphiT}
{
For $c=1$ it holds that for all tripartite biseparable states $\rho_{\text{2-sep}}$
\begin{align}
\Phi_T[{\rho}_{\text{2-sep}}] \geq 0 ,
\end{align}
and this value of $c$ is optimal, i.e. it is the least compatible with the above constraint.}
\end{theorem}

The proof of Theorem \ref{thmphiT} consists in determining the minimal eigenvalue of the operators $T_A \otimes \OI_{\bar{A}}$, 
which turns out to be $-1/2$ for any partition \(A\) (see \cite{partialT}). The detailed proof of Theorem \ref{thmphiT} is presented in Appendix \ref{min eig T}.

The intuition behind this construction rests upon the prerequisite that the negative eigenstates under application of partial transposition for every cut have at least some nonzero overlap. As we will see in the following examples, some states will exhibit this property and the map works in a straightforward way. The additional unitary transformation in eq.(\ref{eq:phidef}) is supposed to systematically map the corresponding negative eigenstates into the same space, opening up the possibility to detect many more states. We will see that such unitary corrections can be readily constructed, but their success depends also on the symmetries exhibited by the state.

To start, let us consider an example where the map directly works. If we consider the state 
\begin{align}
\ket{W}&=\frac{1}{\sqrt{3}}(\ket{001}+\ket{010}+\ket{100}),
\end{align}
we find $\Phi_T[\ketbra{W}{W}]\ngeq 0$ with negative eigenvalue \(1-\frac{2}{\sqrt{3}} \approx -0,15\) and therefore $\Phi_T$ provides a GME criterion. The map  $\Phi_T$ can detect the noisy $W$ state,  \(p  \ket{W}\bra{W} + (1-p)  \frac{\mathbb{I}}{8}\), for  \(p >\frac{11 \sqrt{3}}{16+ 3 \sqrt{3} } \approx 0,90\). Note that the noise resistance of this detection criteria is significantly lower than the actual noise threshold, i.e. the minimum $p$ necessary for the noisy \(W\) state to be GME is \(p=0.4790\) \cite{guhnetaming}. 
{An intuitive reason for the inefficiency of map $\Phi_T$ is that  no attempt is made to map the negative 
eigenstates of the partial transposition with respect to different bipartitions into the same subspace. As we are going to see in  Secs. \ref{sec:diagM} and \ref{sec:choi}, such  procedures provide robust criteria for noisy GHZ states. } \\

The map $\Phi_T$ can be modified to detect the $\ket{\text{GHZ}}= \frac{1}{\sqrt{2}} (\ket{000} + \ket{111})$ state, by choosing $\mathcal{U}^{(A)}=\tilde{\sigma}_x^A$, where 
$\tilde{\sigma}_x^A[\rho]=\prod_{i\in A}{\sigma_x}_i\otimes \I_{\bar{i}}\rho\prod_{i\in A}{\sigma_x}_i\otimes \I_{\bar{i}}$. Let us consider  
\begin{align}\label{eqmapTx}
\Phi_{Tx}[\cdot]=(\tilde{\sigma}_x\circ T_A \otimes \OI_B \otimes \OI_C + \OI_A \otimes \tilde{\sigma}_x\circ T_B \otimes \OI_C  + \OI_A \otimes \OI_B \otimes \tilde{\sigma}_x \circ T_C  +\OI\cdot \text{Tr})[\cdot]
\end{align}
where the positive map $\tilde{\sigma}_x\circ T$ denotes the transposition followed by the application of the unitary Pauli operator $\sigma_x$.
Since we only added a local unitary operation, the map $\Phi_{Tx}$ remains positive on all biseparable states. One calculates $\Phi_{Tx}[\ketbra{\text{GHZ}}{\text{GHZ}}]\ngeq 0$ with negative eigenvalue \(- \frac{1}{2}\) . Therefore, \(\Phi_{Tx}\) is  also a GME-map. 
It furthermore detects the noisy GHZ state, \(p \ket{\text{GHZ}}+ (1-p) \frac{\mathbb{I}}{8}\), for all \(p>\frac{11}{15} \approx 0.73\). 
Though this detection is more robust than the one of the noisy \(W\) by \(\Phi_T\), it is, despite the non-trivial choice of local unitares, 
still far from detecting all the noisy GHZ states that are GME; the condition for the latter being \(p \geq \frac{3}{7}\)\cite{guhnesee}. Section \ref{sec:diagM} is devoted to optimizing this detection criteria by making a more suitable choice of local unitaries.

{Also note that the reason why the map $\Phi_{Tx}$ successfully detects the  $\text{GHZ}$ state, as opposed to $\Phi_T$, is that by 
composing   the transposition with the unitary operation $\sigma_x$ we are exploring the symmetries of the state: note that 
the off-diagonal elements of the $\text{GHZ}$ state remain invariant under the application of partial transposition followed 
by the Pauli matrix $\sigma_x$.}
{Further considerations of symmetry will be made in section \ref{sec:choi}. }

The map $\Phi_{Tx}$ can be generalised for an arbitrary number of parties. Let $A \subset \DE{1,\ldots, n}$ be a proper subset of the parties, we define the map
\begin{align}\label{eqmapTn}
\Phi_{Tx,n}[\cdot]=\de{\sum_{A} \tilde{\sigma}_x^A \circ T_A \otimes \OI_{\bar{A}} +(2^{n-1}-2)\frac{1}{2}\I\cdot \text{Tr}}[\cdot],
\end{align}

\begin{theorem}\label{thmphiTn}
For every $n$-partite biseparable state  $\rho_{2-sep}$ it holds that
\begin{align}
\Phi_{Tx,n}[{\rho}_{2-sep}] \geq 0.
\end{align}
\end{theorem}

The proof of Theorem \ref{thmphiTn} follows in the same line as the proof of Theorem \ref{thmphiT}. 
Moreover, $\Phi_{Tx,n}[\ketbra{\text{GHZ}_n}{\text{GHZ}_n}]\ngeq 0$ for all $n$, with negative eigenvalue \(- \frac{1}{2}\) which implies that \(p \ket{\text{GHZ}} + (1-p) \frac{\mathbb{I}}{2^n}\) is detected for all \( p >\frac{2^{2n-2} -2^n+ 2^{n-1}-1}{2^{2n-1}-1} \stackrel{n \rightarrow \infty}{\rightarrow} 1,\) 
where $\ket{\text{GHZ}_n}=\frac{1}{\sqrt{2}} (\ket{00 \dots 0} + \ket{11 \dots 1})$ is the $n$-partite GHZ state. Though this naive generalization has the great advantage of detecting \(\ket{\text{GHZ}_n}\) 
for any $n$, its noise resistance scales badly with $n$ and this criterion fails to detect a large portion of the $n$-partite GME noisy GHZ states, 
the latter being GME for all \(p> \frac{2^{n-1}-1}{2^n-1}\) \cite{guhnesee}. Section \ref{sec:diagM} will correct for this by modifying this GME 
map in order to detect all $n$-partite (and in fact qudit) GME noisy GHZ states.
Finally note that the map \(\Phi_T\) can also be similarly generalised for $n$ parties, however, already for \(n=4\), the $n$-partite $W$ state is not detected anymore. 
We therefore see that this construction although generating GME-maps, does not provide a noise resistant detection. 
This might have two causes, namely the choice of the positive map 
or the naive choice of $M$. See sections \ref{sec:diagM} or \ref{sec:choi} for less naive constructions.

\subsection{\label{sec:diagM}Optimized transposition criteria}

In this section we will modify the map \(\Phi_{Tx,n}\) of section \ref{sec:trans} by looking at a slightly better choice for the correction map \(M\). This will enable us to detect the GHZ state with an optimal noise resistance. We will choose \(M=(2^{n-1}-2) \text{Diag}\circ \phi\), where here \(\phi= \sum_A \tilde{\sigma}_x^A \circ T_A \otimes \OI_{\bar{A}}\) {and $\text{Diag}[\rho]$ maps $\rho$ to a diagonal matrix with the same elements as $\rho$}.  Hence we look at the map
\begin{equation}
\eta \coloneqq \phi + (2^{n-1}-2) \text{Diag} \circ \phi.
\end{equation}

We also want to project our $n$-qubit state onto a subspace before applying \(\eta\) to it. The subspace we want to project onto is 
\begin{equation}
\text{span} \{\ket{i}\bra{j} \mid \bra{i} \text{GHZ}_{n,\text{cyclic}} \ket{j} \neq 0\},
\end{equation}
 where for \(X^C\coloneqq \mathbb{I} \otimes X^{C_2} \otimes \dots \otimes X^{C_n}\), with \(X\coloneqq \begin{pmatrix} 0&1\\ 1&0 \end{pmatrix}\) and \(C\coloneqq (C_2, \dots, C_n) \in \{0,1\}^{n-1}\), we have 
 \begin{equation}\label{Xstates}
 \text{GHZ}_{n,\text{cyclic}}\coloneqq \sum_{C \in \{0,1\}^{n-1}}X^C \ket{\text{GHZ}_n} \bra{\text{GHZ}_n}(X^C)^{\dagger}.
 \end{equation}
  {The class of states invariant under the projection onto subspace \eqref{Xstates} is also known as $X$-states \cite{Xstates}, and genuine multipartite entanglement in this class of states was fully characterised in Ref.~\cite{gmeXstates}.} 
The notation is more extensively explained in section \ref{sec:choi}. 
{The projector defined in Eq. \eqref{Xstates} can be seen as a projection onto the subspace of GHZ$+$ states, as defined in Ref.~\cite[Eq. (5)]{DC00}. In \cite{DC00} the authors have shown that by local depolarization 
 any $n$-qubit state can be taken to a state diagonal in the GHZ basis.}
As we are going to see in Proposition \ref{prop-main} in section \ref{sec:choi}, it turns out that one can project onto this subspace via a mixture of local unitaries (the assertion holds for any \(d \geq 2\) but for now we are interested in the case \(d=2\)), thus ensuring that no entanglement is created in the process. We denote this projection by \(\mathcal{X}_n\). Our map of interest is thus \(\eta \circ \mathcal{X}_n\). 

\begin{theorem}\label{thm_eta}
\(\eta \circ \mathcal{X}_n [\rho_{\text{2-sep}}] \geq 0\), for all n-qubit biseparable states \(\rho_{\text{2-sep}}\).
\end{theorem}

\begin{proof}
The proof relies on the following fact:

\begin{lemma}\label{lemmaODT}
\[\text{OD}\circ \phi [\mathcal{X}_n [\rho]] = (2^{n-1}-1) \text{OD}[\tilde{\sigma}_x^A \circ T_A \otimes \OI_{\bar{A}} [\mathcal{X}_n[\rho]]], \quad \forall \rho, \forall A,\]
where \(\text{OD}[X]\coloneqq X- \text{Diag} [X]\) {and \(\phi= \sum_A \tilde{\sigma}_x^A \circ T_A \otimes \OI_{\bar{A}}\)}.
\end{lemma}


The intuition behind this fact is that for $X$-states, the off-diagonal elements are equally permuted by partial transposition in subsystem $A$ and $\sigma_x$ flips in the same subsystem, i.e. all their off-diagonal elements are invariant under application of \(\tilde{\sigma}_x^A \circ T_A \otimes \OI_{\bar{A}}\).

Now let us consider an $n$-partite state biseparable with respect to partition $A|\bar{A}$: $\rho=\rho_A \otimes \rho_{\bar{A}}$. By projecting $\rho$ 
onto the above subspace the resultant state $\tilde{\rho}\coloneqq\mathcal{X}_{n}(\rho)$ is still biseparable with respect to partition $A|\bar{A}$ since the projection $\mathcal{X}_{n}$ is a separable operation.  Analysing the map
$\eta$ applied to $\tilde{\rho}$ gives:
\begin{align}
 \eta[\tilde{\rho}]=& \text{Diag}[\phi[\tilde{\rho}]]+\text{OD}[\phi[\tilde{\rho}]] +(2^{n-1}-2) \text{Diag}[\phi[\tilde{\rho}]]\nonumber\\
 =&\text{Diag}[\tilde{\sigma}_x^A  \circ T_A \otimes \OI_{\bar{A}}[\tilde{\rho}]]+\sum_{B \neq A} \text{Diag}[\tilde{\sigma}_x^B  \circ T_B \otimes \OI_{\bar{B}}[\tilde{\rho}]]\nonumber\\
& +(2^{n-1}-1) \text{OD}[\tilde{\sigma}_x^A  \circ T_A \otimes \OI_{\bar{A}}[\tilde{\rho}]]\nonumber\\
&+(2^{n-1}-2) \text{Diag}[\tilde{\sigma}_x^A  \circ T_A \otimes \OI_{\bar{A}} [\tilde{\rho}]]\\
&+(2^{n-1}-2) \sum_{B \neq A}\text{Diag}[\tilde{\sigma}_x^B  \circ T_B \otimes \OI_{\bar{B}} [\tilde{\rho}]]\nonumber\\
=&\overbrace{\text{Diag}[\tilde{\sigma}_x^A  \circ T_A \otimes \OI_{\bar{A}}[\tilde{\rho}]]+(2^{n-1}-2) \text{Diag}[\tilde{\sigma}_x^A  \circ T_A \otimes \OI_{\bar{A}}[\tilde{\rho}]]+(2^{n-1}-1)\text{OD}[\tilde{\sigma}_x^A  \circ T_A\otimes \OI_{\bar{A}}[\tilde{\rho}]]}^{ = (2^{n-1}-1) \tilde{\sigma}_x^A  \circ T_A \otimes \OI_{\bar{A}}[\tilde{\rho}] \geq 0}\nonumber\\
 &+ (2^{n-1}-1)\sum_{B\neq A}\underbrace{\text{Diag}[\tilde{\sigma}_x^B  \circ T_B \otimes I_{\bar{B}}[\tilde{\rho}]]}_{\geq 0} \nonumber\\
 \geq& 0,\nonumber
 \end{align}
 where in the second step we have used Lemma \ref{lemmaODT}.

The proof of Theorem \ref{thm_eta} for an arbitrary biseparable state $\rho$ follows from the linearity of the map $\eta \circ \mathcal{X}_n$.
\end{proof}

The map \(\eta \circ \mathcal{X}_n\) furthermore detects the $n$-qubit GHZ state but now with a improved noise resistance. Indeed for all \(n \geq 2\),
\begin{equation}
\eta \circ \mathcal{X}_n [ p \; \text{GHZ}_n + (1-p) \; \frac{\I}{2^n} ] \ngeq 0 \qquad \forall\ 1 \geq p > \frac{2^{n-1}-1}{2^n-1}.
\end{equation} 

For \(n=3\) for example this means \(1 \geq p > \frac{3}{7}\), meaning \(\eta \circ \mathcal{X}_3\) optimally detects the 3-qubit GHZ state. In fact, it reproduces the necessary and sufficient conditions first presented in Ref.~\cite{guhnesee} for all GHZ-diagonal states and any number of qubits $n$. This has the added benefit that the map is simple to apply and invariant under many local unitary operations, detecting a multitude of states using the same criterion.

\subsection{Reduction-based GME criterion}

Similarly to the transposition based criterion we can construct a GME criterion for tripartite states based on the reduction map. The reduction map acting on
$\rho \in \mathcal{P}(\mathbb{C}^d)$ is defined as
\begin{align}
 R(\rho)\coloneqq \frac{1}{d-1}\de{\text{Tr}(\rho)\I-\rho}.
\end{align}

The corresponding GME map is of the form
\begin{align}
\Phi_{R}[\cdot]=(R_{A}\otimes I_{BC}+ I_{A}\otimes R_{B}\otimes I_C + I_{AB}\otimes R_{C} +c\I \cdot\text{Tr})[\cdot].
\end{align}
We are now ready to prove the following result.

\begin{theorem}\label{thmreduction}
For $c=\frac{2}{d}$ it holds that for all tripartite biseparable state
$\rho_{\text{2-sep}}\in \mathcal{P}(\mathbb{C}^d \otimes \mathbb{C}^d\otimes \mathbb{C}^d)$
\begin{align}
\Phi_{R}[{\rho}_{\text{2-sep}}] \geq 0 .
\end{align}
\end{theorem}

The proof of Theorem \ref{thmreduction} follows in the same way as the transposition based criteria, where the crucial step consists in determining 
the minimum eigenvalue of the reduction map. The details are presented in Appendix \ref{min eig R}.

Now we want to evaluate the action of the map $\Phi_{R}$ on GHZ state of dimension $d$, i.e. $\ket{GHZ(d)}=\frac{1}{\sqrt{d}}\sum_{i=0}^{d-1}\ket{iii}$. In what follows, let us denote by $GHZ(d)$ the projector onto $\ket{GHZ(d)}$. A straightforward calculation shows that
\begin{align}
\Phi_{R}[GHZ(d)]\, =\, \frac2d\, \I + \frac{1}{d(d-1)}\de{ E_{AB}\otimes \I_{C} + E_{AC}\otimes \I_{B} + \I_{A}\otimes E_{BC} } - \frac{3}{d-1}\, GHZ(d)\, ,
\end{align}
where we defined
\begin{align}
E_{AB} \coloneqq \sum_{i=1}^{d} \ketbra{i}{i}_{A}\otimes\ketbra{i}{i}_{B}
\end{align}
and analogously for $E_{AC}, E_{BC}$. The above operator $\Phi_{R}[GHZ(d)]$ can be easily diagonalised by observing that all the addends composing its expression commute with each other. The minimal eigenvalue corresponds to the eigenvector $\ket{GHZ(d)}$ and evaluates to $-1/d$. 

The noise resistance can also be easily obtained 
\begin{align}
\Phi_{R}\De{p\, GHZ(d) + (1-p)\frac{\I}{d^3}}\ngeq 0 \qquad \forall \ p>1-\frac{d^2}{3(d^2+1)}\, .
\end{align}
Note that the above lower bound on the GME threshold tends to $2/3$ as $d\rightarrow\infty$, while it evaluates to $11/15$ for $d=2$. This latter value corresponds to the noise resistance we obtained for the modified lifted partial transposition \eqref{eqmapTx} on the $3$-qubit GHZ state. This is no coincidence, but a consequence of the fact that reduction map and the partial transposition are unitarily equivalent for a two-level system. Of course, one can generalize this criterion for \(n>3\) as has been done in the partial transpose case in section \ref{sec:trans}, not forgetting to take care of the fact that the minimal eigenvalue depends on the dimension of the space the map is applied to. However, already for \(n=3\) the noise resistance detection of the map on the GHZ state is not optimal, which shows that at least
for this state this is a less promising approach than transposition (in the qubit case) or choi map (in the qudit case), see section \ref{sec:choi}, which is why we now continue to explore further maps. 


%
%
%

\subsection{Breuer-Hall map-based GME criterion}

In this section we consider an indecomposable map: the Breuer-Hall map introduced in \cite{Breuer,Hall}. Non-decomposable maps 
are the ones which cannot be written as the sum of a completely positive map and the composition of a completely positive map with transposition, therefore in the bipartite case these are the maps which can detect PPT bound entanglement.

The Breuer-Hall map is defined for even dimensional systems with $d\geq4$ as:
\begin{align}
\mathcal{B}\, =\, \frac{\I\Tr - I - \mathcal{V}T}{d-2}\, ,
\end{align}
where $\mathcal{V}$ represents the application of a skew-symmetric unitary\footnote{An skew-symmetric unitary matrix is a unitary matrix such that $V^T=-V$.} (note that skew-symmetric unitary matrices only exist in even dimension).

Defining the tripartite GME map 
\begin{align}\label{eq:BH3}
\Phi_{\mathcal{B}}[\cdot]=(\mathcal{B}_{A}\otimes I_{BC}+ I_A \otimes \mathcal{B}_{B}\otimes I_C+I_{AB}\otimes \mathcal{B}_{C} +c\I\cdot \text{Tr})[\cdot].
\end{align}
we have the following result.

 \begin{theorem}\label{thmBH}
For $c=\frac{2}{d}$ it holds that for all tripartite biseparable states $\rho_{\text{2-sep}}\in \mathcal{P}(\mathbb{C}^d \otimes \mathbb{C}^d\otimes \mathbb{C}^d)$
\begin{align}
\Phi_{\mathcal{B}}[{\rho}_{\text{2-sep}}] \geq 0 .
\end{align}
\end{theorem}

The proof of Theorem \ref{thmBH} is presented in Appendix \ref{min eig BH}.


As in the previous section we evaluate this map on the GHZ state of dimension $d$ and find
\begin{equation}
\Phi_{\mathcal{B}}[\text{GHZ}(d)]= \frac{d-1}{d-2} \Phi_R[\text{GHZ}(d)] - \frac{2}{d(d-1)} \I - \frac{1}{d(d-2)} (F_A+F_B+F_C), 
\end{equation}
where 
\begin{equation}
F_A := \sum_{ij} V_A \ket{jii} \bra{ijj} V_A^{\dagger}
\end{equation}
and analogously for \(F_B\) and \(F_C\). Choosing V as 
\begin{equation}
V= \begin{pmatrix}0 & \I \\ - \I & 0 \end{pmatrix},
\end{equation}
 \(\Phi_{\mathcal{B}}[\text{GHZ}(d)]\) can again be diagonalised by observing that the addends commute. The minimal eigenvalue of the latter is also \(-\frac{1}{d}\), \(\ket{\text{GHZ}(d)}\) being once more an eigenvector corresponding to this eigenvalue. The noise resistance is also identical to the previous section. The generalization to the case \(n>3\) can be carried out as in the previous section, the remarks made there being also valid in this case.
 
 Even though the map $\mathcal{B}$ is a non-decomposable one, it is very unlikely that the lifting map \eqref{eq:BH3} can detect GME states which are PPT with respect to all bipartitions, and the reason is that the identity compensation map is too naive.
In Section \ref{sec:diagM} we have seen that by exploring the symmetries of the positive map in consideration we can design a
better compensation map which may lead to very strong criterion.  
Exploring the symmetries of the Breuer-Hall map may lead to a strong GME criterion that can detect states which are PPT with respect to
all bipartitions. We leave it as an open point for further investigation. In the next Section we consider another indecomposable map, the Choi map,
and explore its symmetry constructing a very strong criterion which can detect GME states in a robust way.

\subsection{\label{sec:choi}Choi-based GME criteria}

The transposition-based GME criteria of section \ref{sec:trans} illustrates the general idea of our method and provides a GME criteria for any number of parties. 
However the application of Theorem \ref{thmphiT} to noisy states does not provide a robust criterion.
Modifying the correction term $M$ as in section \ref{sec:diagM} delivered an improved criteria for $n$-qubit GHZ states that was even revealed to be optimal for any $n$. 
We now want to optimise the criterion for $n$-qudit states by defining a GME-map based on the non-completely positive Choi-map \cite{choi}. 
{An important step in the construction of our GME-map is to} 
explore the symmetries of the Choi-map in order to add a correction term reflecting them. As a result we obtain a map which is very robust to detect GHZ-like states for any dimension and number of parties.

The construction of the Choi-based GME criterion consists, as in section \ref{sec:diagM}, of two steps. The first step is a projection onto a subspace of GHZ-like states by a mixture of local unitary operations, and the second step is the application of a non-positive map based on the Choi-map.

First of all, let us introduce some notation. A suitable generalisation of the GHZ state to $n$-parties and $d$-dimensions can be defined as
\begin{align}\label{eq.ghznd}
\ket{GHZ_n^d}=\frac{1}{\sqrt{d}}(\ket{00\ldots 0}+\ket{11\ldots1}+\ldots +\ket{(d-1)(d-1)\ldots(d-1)}).
\end{align}
And we denote by GHZ-like state an $n$-partite $d$-dimensional state that differs from \eqref{eq.ghznd} by local unitary operations. 
We will be interested in the subspace generated by cyclic permutations of \(\ket{\text{GHZ}_n^d}\), {which we can think of as the subspace that leaves invariant the family of states that generalises $X$-states to higher dimensions}. In order
to describe our subspace we are going to use the shift operator, the generator of cyclic permutations:
\begin{align}\label{eqXd}
X_d\coloneqq\begin{pmatrix} 
0&1&0&\cdots &0\\ 
0&0&1&\ddots&\vdots\\
\vdots &\ddots &\ddots &\ddots&0\\
0&\cdots&0&0&1\\
1&0&\cdots&0&0 \end{pmatrix}\,,
\end{align}
and consider the  matrix
\begin{align}
\text{GHZ}_{n,\text{cyclic}}^{d}\coloneqq \sum_{\mathbf{C} \in \{0,1,\dots, d-1\}^{n-1}} X^{\mathbf{C}}_d \ket{{GHZ_n^d}} \bra{{GHZ_n^d}} (X^{\mathbf{C}}_d)^{\dagger},
\end{align}
where for \(C\coloneqq (C_2, \dots, C_n) \in \{0,1\}^{n-1}, \quad X^C\coloneqq \I \otimes X^{C_2} \otimes \cdots \otimes X^{C_n}\).
The span of \(\{\ketbra{i}{j} \mid \bra{i} \text{GHZ}_{n,\text{cyclic}}^d \ket{j} \neq 0\}\) defines a subspace that we denote by \(\{\text{GHZ}_{n,\text{cyclic}}^d\}\). Note that \(\{\text{GHZ}_{n,\text{cyclic}}^2\}\) is the subspace considered in section \ref{sec:diagM}.

We now want to show that we can project any $n$-partite state \(\rho \in \mathcal{P} (\mathbb{C}^d\otimes \dots \otimes \mathbb{C}^d)\) onto
\(\{\text{GHZ}_{n,\text{cyclic}}^d\}\) by a mixture of local unitary operations.
In order to prove it we construct our projection by making use of the  \textit{clock-matrices}:
\begin{align}   
Z_k= \begin{pmatrix} 
 1 &0 &\cdots&0\\ 
 0& e^{2 \pi i k/d} &\ddots&0 \\  
 0& 0& e^{2 \pi i k  2/d} &\vdots\\  
 \vdots& \ddots &\ddots &0\\ 
 0&\cdots & 0& e^{2 \pi i (d-1)k/d} 
 \end{pmatrix},
 \end{align}
 where \(k \in \{0, 1, \dots, d-1\} \).
 
Furthermore, let \(\rho \in \mathcal{P} (\mathbb{C}^d\otimes \dots \otimes \mathbb{C}^d)\) be an $n$-qudit state. For \(i \in \{1, \dots, n\}\) we define
\begin{align}  
  h^d_{n,i} (\rho) \coloneqq \frac{1}{d} \sum_{k=0}^{d-1} Z_k^1 (Z_k^i)^{\dagger} \;\rho \;(Z_k^1)^{\dagger} Z_k^i,
\end{align}
where \(Z_k^i \coloneqq \I \otimes \dots \otimes \I \otimes \stackrel{\downarrow i^{\text{th}}}{Z_k} \otimes \I \otimes \dots \otimes \I\).

And we finally construct the operator:
\begin{align}\label{eqmapproj}
 \mathcal{X}_{n}^d[\cdot]\coloneqq \circ_{j=2}^n h_{n,j}^d[\cdot].
\end{align}

The following proposition states the equivalence of the the map \eqref{eqmapproj} with the projection on 
the subspace \(\{\text{GHZ}_{n,\text{cyclic}}^d\}\).

\begin{proposition}\label{prop-main}
For all \(n \in \mathbb{N},\; d \geq 2\), $\mathcal{X}_{n}^d [\rho]$ projects any $n$-partite $d$-dimensional state $\rho$ into the subspace  \(\{\text{GHZ}_{n,\text{cyclic}}^{\;d}\}\).
\end{proposition}

The proof of Proposition \ref{prop-main} can be found in Appendix \ref{Aprojection}.

Now that we have a way to project any $n$-qudit state into the subspace spanned by \textit{GHZ-like cyclically permuted} states,
we can continue to discuss the positive maps employed. 

\begin{definition}
 The Choi-map $\Lambda$ of dimension $d$ is defined as follows
 \begin{align}
  \Lambda [\rho]=2 \text{Diag}[\rho]+\sum_{j=1}^{d-2}X_d^j\text{Diag}[\rho]{X_d^j}^{\dagger}-\rho,
 \end{align}
where $X_d$ is the shift operator defined in Eq. \eqref{eqXd}.
\end{definition}

We next want to look at
\begin{align}\label{eqphi}
{\phi}=\sum_{A} \Lambda_A \otimes \OI_{\bar{A}}, 
\end{align}
where $A \subset\{1,\ldots,n\}$
and $\Lambda_A$ denote the Choi-map applied to 
the subsystems contained in $A$. With the above we can consider the map
\begin{align}\label{eq_mu}
\mu[\rho] &= \phi[\rho] +(2^{n-1}-2)\left[\text{Diag}[\phi[\rho]]-\sum_{A} \text{Diag}[\rho] \right].
\end{align}

Now we are ready to state the main result of this section.

\begin{theorem}\label{thm_choigme}
\(\mu \circ \mathcal{X}_{n}^d[\rho_{2-sep}] \geq 0\).
\end{theorem}

Note that in order to express \(\mu \circ \mathcal{X}_n^d\) in the form of \(\Phi_{\text{GME}}\) of equation \eqref{eq:phidef} in section \ref{sec.results}, we should set \(\mathcal{U}^{(A)}= \mathcal{X}_n^d\) for every \(A\) 
and \(M=(2^{n-1}-2)\left(\text{Diag}[\phi[\mathcal{X}_n^d[\rho]]]-\sum_{A} \text{Diag}[\mathcal{X}_n^d[\rho]] \right)\). Theorem \ref{thm_choigme} gives us a sufficient condition for a state to be GME: if upon projecting an $n$-partite state $\rho$ into
the {GHZ-like cyclically permuted} subspace and applying the map $\mu$ results into a negative eigenvalue, 
one can assure that the state $\rho$ is genuinely $n$-partite entangled.

\begin{proof}
The proof is analogous to the one of Theorem \ref{thm_eta} of section \ref{sec:diagM}. Indeed, we make use of the following Lemma that will be proven in Appendix \ref{Aprooflemma}.

\begin{lemma}\label{lemmaOD}
 \begin{align}
\text{OD}\circ \phi[\mathcal{X}_{n}^d[\rho]] = (2^{n-1}-1)\text{OD}[\Lambda_A\otimes I_{\bar{A}}[\mathcal{X}_{n}^d[\rho]]]\;\; \forall \, \rho, \forall \, A,
 \end{align}
 where $\text{OD}[X]=X-\text{Diag}[X]$.
\end{lemma}

Now consider an $n$-partite state biseparable with respect to partition $A$, $\rho=\rho_A \otimes \rho_{\bar{A}}$. We note that $\tilde{\rho}\coloneqq\mathcal{X}_{n}^d(\rho)$ is still biseparable with respect to the partition \(A \lvert \bar{A}\) and hence,
like in section \ref{sec:diagM}:
\begin{align}
\begin{split}
 \mu[\tilde{\rho}]=& \text{Diag}[\phi[\tilde{\rho}]]+\text{OD}[\phi[\tilde{\rho}]] +(2^{n-1}-2)\left(\text{Diag}[\phi[\tilde{\rho}]]-\sum_A\text{Diag}[\tilde{\rho}] \right)\\
 =&\text{Diag}[\Lambda_A \otimes \OI_{\bar{A}}[\tilde{\rho}]]+\sum_{B \neq A} \text{Diag}[\Lambda_B \otimes \OI_{\bar{A}}[\tilde{\rho}]]\\
& +(2^{n-1}-1) \text{OD}[\Lambda_A \otimes \OI_{\bar{A}}[\tilde{\rho}]]\\
&+(2^{n-1}-2) \text{Diag}[\Lambda_A\otimes \OI_{\bar{A}} [\tilde{\rho}]]\\
&+(2^{n-1}-2) \sum_{B \neq A}\text{Diag}[\Lambda_B\otimes \OI_{\bar{B}} [\tilde{\rho}]]\\
&- (2^{n-1}-1) \sum_{B \neq A }\text{Diag}[\tilde{\rho}]\\
=&\overbrace{\text{Diag}[\Lambda_A \otimes \OI_{\bar{A}}[\tilde{\rho}]]+(2^{n-1}-2) \text{Diag}[\Lambda_A \otimes \OI_{\bar{A}}(\tilde{\rho})]+(2^{n-1}-1)\text{OD}[\Lambda_A\otimes \OI_{\bar{A}}[\tilde{\rho}]]}^{ = (2^{n-1}-1) \Lambda_A \otimes \OI_{\bar{A}}[\tilde{\rho}] \geq 0}\\
 &+ (2^{n-1}-1)\sum_{B\neq A}(\underbrace{\text{Diag}[\Lambda_B\otimes I_{\bar{B}}[\tilde{\rho}]]}_{\geq \text{Diag} (\tilde{\rho})}-\text{Diag}[\tilde{\rho}])\\
 \geq& 0,
\end{split}
 \end{align}
 where in the second step we have used Lemma \ref{lemmaOD}.

The proof of Theorem \ref{thm_choigme} for an arbitrary biseparable state $\rho$ follows again from the linearity of the map $\mu \circ \mathcal{X}_n^d$.
\end{proof}

We now want to look at some states detected by \(\mu \circ \mathcal{X}_n^d\). Our first example is the noisy $n$-partite $d$-dimensional GHZ state:
\begin{align}
 \rho_{\text{GHZ}}^{n,d}=\alpha \ketbra{GHZ_n^d}{GHZ_n^d}+(1-\alpha) \frac{\I}{d^n}.
\end{align}

Applying our map to this state gives us the critical value $\alpha_c$ for which genuine multipartite entanglement is detected by it:
\begin{align}
\alpha_c=\frac{(d-2) (2^{n-1}-1) + 1} {(d-2) (2^{n-1}-1) + 1 + (d-2) d^{n-1}},
\end{align}
which for $d>2$ fixed goes to zero exponentially with the number of parties, as indeed \(\alpha_c= o((\frac{2}{d})^n) \quad (n \rightarrow \infty )\). 
This means that for \(d >2\) 
and $n$ big enough, our map detects the noisy $n$-qudit GHZ state with up to almost \(100 \%\) white noise.

Furthermore, since our map is based on the non-decomposable Choi-map it can also detect
genuine multipartite entanglement in systems which are positive under partial transposition with respect to any bipartition. 
To illustrate that we can consider the family of 3-qutrit states $\DE{\rho(\lambda)}_{\lambda \in \mathbb{R}^+}$ introduced in Ref. \cite{hubersengupta}, see also Appendix \ref{app.qutritppt} for their definition.

The states  $\rho(\lambda)$ have the property to be invariant under partial transposition, therefore partial transposition based criteria cannot even detect bipartite entanglement. By applying our map we recover  the results obtained in Ref.  \cite{hubersengupta}, showing that the states  $\rho(\lambda)$ are GME for $0\leq \lambda < \frac{1}{3}$, which to our knowledge is still the best known detection range. Setting \(\lambda = \frac{1}{9}\), the noisy state
\begin{align}
\rho_{noise}(p)=p\frac{\I}{27}+(1-p)\rho\de{\frac{1}{9}}
\end{align}
is detected with our criteria with white noise up to $\frac{9}{179}\approx 5 \%$. Apart from Refs.~\cite{hubersengupta,relaxations} and Ref.~\cite{MarcoP} we are not aware of any GME detection technique powerful enough to achieve this feature. \\

We have just seen that our Choi based criteria is able to detect the noisy $n$-qudit GHZ state with noise resistance increasing with $n$. 
It can also detect sates that are PPT with respect to every cut with 
the best known white noise resistance. It is thus fair to say that our scheme enabled us to derive a strong criterion. 
The key {for obtaining such a criterion was to explore the symmetries of the Choi-map,} 
namely that it acts the same on every off-diagonal element of a state, and to be able to project in a subspace,
\(\{\text{GHZ}_{n, \text{cyclic}}^d\}\), exhibiting this symmetry without creating any entanglement.

We believe that this strategy can be applied to other maps, such as the Breuer-Hall map to cite only one. 
{One then is left with the task of identifying the symmetry that characterizes the map, and more difficultly to find a way of making a projection
  into a subspace exhibiting this symmetry without creating any entanglement.}
{In order to  shed some light to the possibility of} 
accomplishing the second step of this process,  we include an alternative proof of Proposition \ref{prop-main} in Appendix \ref{Aprojection}. Indeed, since the alternative proof is based on group theory, we believe it will be easier to adapt it to a map having another symmetry.

\section{Discussion}\label{sec.discussion}
In this manuscript we investigate the generalisation of the bipartite concept of positive maps to GME-maps in order to detect 
genuine multipartite entanglement, therefore closing a gap in the set of available tools for revealing GME. We introduce a general method of constructing GME-maps based on positive but non-completely positive maps. We have furthermore illustrated our construction method by generalising the paradigmatic PPT criterion to the multipartite case. A systematic way of generalising any bipartite map has furthermore been introduced and lifting of the Reduction as well as the Breuer-Hall map have been carried out using it. As a first application, we showed that our approach provides a novel solution to the problem of determining the bi-separability threshold of GHZ mixtures. Moreover, by exploring the symmetries of the Choi-map we were able to design a very robust criterion (for GHZ-like states), which can recover the results of some of the strongest known criteria for 3-qutrit and even reveal entanglement in states that are PPT with respect to every partition. We are therefore strongly convinced that this construction could be a first step on a path to a better understanding of GME, even if only for the sole purpose to derive new GME-witnesses by lifting known positive maps to GME-maps.\\

There are, however, still many open questions concerning our construction. In the bipartite case, one of the upsides of the 
approach based on positive maps is that in some cases it is possible to achieve invariance under all local unitary for
the resulting criteria (as with PPT). This also directly leads to the fact that every entangled pure state can easily be detected by this one simple map. 
Now the question that arises here is analogous: can there be a single multipartite map that reveals GME for all locally unitarily related (or even all) GME pure states? While our maps are performing very strongly in terms of noise resistance and can even detect states which are PPT with respect to every bipartition, they are not invariant under all local unitaries and will probably fail to reveal GME of all GME pure states. The main challenge appears due to the fact that the negative eigenstates under every partition need to mutually overlap, in order for the construction to work. The lack of a unique Schmidt decomposition for multipartite systems and the hardness of computing the tensor rank make this a tough challenge, that we were only able to overcome in scenarios with a high degree of symmetry.

Moreover, from a more technical point of view, in order to develop new examples of the maps we propose it would be crucial to have a general method 
of deriving a valid, nontrivial, compensation map $M$, without having to first project into a subspace of the multipartite system (as in this way we expect to miss many important features and focus only on a small niche of the richness of GME correlations).

Finally, an important open point is whether every GME state can be detected by a GME-map lifted from positive maps. For the GME-witnesses a similar construction method was indeed recently shown to suffice to reveal the entanglement of any GME-state \cite{relaxations}.
An affirmative answer would show that among all the GME-maps, one would only need to consider the ones lifted from positive maps to fully characterize genuine multipartite entanglement, thus greatly simplifying the search for such maps.

\section{Acknowledgement}
FC and MH acknowledge funding from Swiss National Science Foundation (AMBIZIONE PZ00P2$\_$161351). M. H. furthermore acknowledges funding from the Austrian Science Fund (FWF) through the START project Y879-N27. GM acknowledges financial support from Funda\c{c}\~{a}o de Amparo \`{a} Pesquisa do Estado de Minas Gerais (FAPEMIG), NWO VIDI and ERC Starting Grant. LL acknowledges financial support from the European Union under the European Research Council (AdG IRQUAT No. 267386), the Spanish MINECO (Project No. FIS2013-40627-P), and the Generalitat de Catalunya (CIRIT Project No. 2014 SGR 966).

\appendix

\section{\label{app:W_GME}Deriving \(W_{\text{GME}}\)}

{In this section we want to derive the form of \(W_{\text{GME}}\) as presented in \eqref{eq.W}}. We know from Theorem III.1. of \cite{relaxations} that for any GME-state \(\rho_{\text{GME}}\) there exists an operator \(Q\) as well as operators \(\{T_A\}_A\) such that each element of the set of weakly optimal bipartite witnesses \(\{W_A\}_A\) defined by \(W_A = Q+T_A\) detects \(\rho_{\text{GME}}\) and such that

\begin{equation}
W\coloneqq Q + \sum_A [T_A]_+,
\end{equation}
where \([T_A]_+\) denotes the projection onto the positive semidefinite cone of \(T_A\), is a multipartite witness detecting \(\rho_{\text{GME}}\). Furthermore by Corollary III.2 of \cite{relaxations}, 
for each $A$, the \(W_A\) can
be assumed to be of the form \(W_A= \Lambda_A^* \otimes \OI_{\bar{A}} [ \ket{\psi_A} \bra{\psi_A}]\), where
for each $A$, \(\Lambda_A\) is a positive map detecting \(\rho_{\text{GME}}\) and \(\ket{\psi_A}\) is chosen such that 
\(\bra{\psi_A} \Lambda_A \otimes \OI_{\bar{A}} ( \rho_{\text{GME}}) \ket{\psi_A} <0\). Then by defining
\begin{equation}
M_{\{\Lambda_A, \ket{\psi_A}\}_A} \coloneqq \sum_A [ (2^{n-1}-1) [T_A]_+ - T_A] \geq 0;
\end{equation}
we have that 
\begin{align*}
W_{\text{GME}}&\coloneqq \sum_A \Lambda_A^* \otimes \OI_{\bar{A}} [ \ket{\psi_A} \bra{\psi_A}] + M_{\{\Lambda_A, \ket{\psi_A}\}_A}\\
&= \sum_A (Q+T_A)+ \sum_A ((2^{n-1}-1) [T_A]_+ - T_A)\\
&= (2^{n-1}-1) W;
\end{align*}
proving that there exists a GME-witness detecting \(\rho_{\text{GME}}\) of the desired form, namely \(W_{\text{GME}}\).

\section{Minimal output eigenvalue of positive maps} \label{min eig pos}

Given a positive, trace-preserving map $\zeta$ acting on $d\times d$ matrices, let us consider its {\it minimal output eigenvalue} as defined by
\begin{equation}
\mu(\zeta)\, \equiv\, - \min_{\rho} \left\{\text{EV}_{\min} \left((\zeta\otimes I)[\rho]\right)\right\}\, , \label{mu}
\end{equation}
where the minimisation is over all normalised bipartite states $\rho\in\mathcal{P}\left(\mathds{C}^d\otimes \mathds{C}^n \right)$ with $n$ an arbitrary positive integer, and $\text{EV}_{\min}$ denotes the minimal eigenvalue. The reason for considering this rather than the minimal eigenvalue itself is that for positive but not completely positive maps (the ones we are interested in) the former is a positive quantity. First note that the following elementary properties hold true:
\begin{enumerate}
\item $\mu$ is convex;
\item \(\mu(\zeta)= - \min_{\ket{\psi}} \left\{\text{EV}_{\min} \left((\zeta\otimes I)[\ket{\psi}\!\!\bra{\psi}]\right)\right\}=\max_{\ket{\psi}} \left\{-\text{EV}_{\min}\left((\zeta\otimes I)[\ket{\psi}\!\!\bra{\psi}]\right)\right\}\);
\item $\zeta$ is completely positive iff $\mu(\zeta) = 0$;
\item if $\phi$ is completely positive, unital and trace-preserving, then $\mu(\phi\zeta) \leq \mu(\zeta)$.
\end{enumerate}

Let us comment on the above claims. The convexity of $\mu(\zeta)$ in $\zeta$ follows immediately from its definition. Moreover, the minimisation in \eqref{mu} can be restricted to rank-one projectors as follows from the fact that any state $\rho$ can be expressed as a convex combination of pure states. Next, complete positivity of $\zeta$ implies by its very definition that $\mu(\zeta)\leq 0$. However, choosing a companion system with dimension $n>d$ shows that this bound can be achieved. Finally, the last property is a straightforward calculation.

In our approach, the minimal output eigenvalue of a positive map $\zeta$ is half the coefficient $c$ of the compensation map $\I \cdot \text{Tr}$ that one needs to add in order to make $\zeta_A\otimes \OI_{BC}+\OI_A\otimes \zeta_B\otimes \OI_C+\OI_{AB}\otimes \zeta_C$ positive on biseparable states. The paradigmatic example of this method is Theorem \ref{thmphiT}, but these ideas can be pushed further to include the more general case of a higher number of parties (see Theorem \ref{thmphiTn}). Since this minimal output eigenvalue problem is so important to solve a number of illustrative examples, we devote this appendix to the computation of $\mu(\zeta)$ for some notable positive maps.

\subsection{Partial transpose} \label{min eig T}

It is well known that $\mu(T)=1/2$ \cite{partialT}, but we include a proof here. 

\begin{lemma} \label{mu T}
The partial transposition $T$ satisfies $\mu(T)=1/2$.
\end{lemma}

\begin{proof}
We already know that it is sufficient to examine the minimal output eigenvalue of a pure state $\ket{\psi}$. Hence, take $\ket{\psi}\in\mathds{C}^d\otimes\mathds{C}^n$ whose Schmidt decomposition is $\ket{\psi}=\sum_i c_i \ket{ii}$. It is easy to see that $(T\otimes \OI)(\ketbra{\psi}{\psi})$ preserves the subspaces $W_{ij}\coloneqq \text{Span}\{\ket{ij},\ket{ji}\}$. Note that we have here for simplicity taken the partial transpose with respect to the Schmidt basis. If the partial transpose was taken with respect to another product basis, the invariant subspace would be $W_{ij}\coloneqq \text{Span}\{\ket{i^*j},\ket{j^*i}\}$, the complex conjugation being taken in the same basis. On the one hand, for $i\neq j$ we obtain
\begin{equation}
(T\otimes \OI)[\ketbra{\psi}{\psi}]\big|_{W_{ij}}\, =\, \begin{pmatrix} 0 & c_i c_j \\ c_i c_j & 0 \end{pmatrix} ,
\end{equation}
where the basis chosen for the above matrix representation is naturally $\{\ket{ij},\ket{ji}\}$. On the other hand, all vectors $\ket{ii}$ are eigenvectors of $(T\otimes \OI)(\ketbra{\psi}{\psi})$ with eigenvalue $c_i^2$. From these identities it is clear that each subspace $W_{ij}$ ($i\neq j$) contributes with a negative eigenvalue $-c_i c_j$, and that these are the only negative eigenvalues of $(T\otimes \OI)(\ketbra{\psi}{\psi})$. Since $\sum_i c_i^2=1$, it is straightforward to verify that $c_i c_j\leq \frac12$, with the upper bound being achieved by $c_1=c_2=\frac{1}{\sqrt{2}}$ and $c_i=0$ for $i>2$.
\end{proof}

We now want to prove Theorem \ref{thmphiT} of the main text.

\begin{proof}[Proof of Theorem \ref{thmphiT}]
Employing Lemma \ref{mu T} we see that
\begin{align*}
\text{EV}_{\min}&(\Phi_T[\rho_{\text{2-sep}}])= \text{EV}_{\min}\left(\sum_A T_A \otimes \OI_{\bar{A}} \left[ \sum_{A'} \sum_i p_{A'}^i \cdot \rho_{A'}^i \otimes \rho_{\bar{A'}}^i\right] + \I\right)\\
=&\text{EV}_{\min}\left(\sum_{A'}\sum_i \left( \sum_{A=A'} T_{A'} \otimes \OI_{\bar{A'}} \left[ p_{A'}^i \cdot \rho_{A'}^i \otimes \rho_{\bar{A'}}^i\right]+\sum_{A\neq A'} T_A \otimes \OI_{\bar{A}} \left[ p_{A'}^i \cdot\rho_{A'}^i \otimes \rho_{\bar{A'}}^i\right]\right) + \I\right)\\
\geq& \sum_{A'}\sum_i \left[ \sum_{A=A'} p_{A'}^i  \underbrace{\text{EV}_{\min}\left(T_{A'} \otimes \OI_{\bar{A'}} \left[ \rho_{A'}^i \otimes \rho_{\bar{A'}}^i\right]\right)}_{\geq 0}+\sum_{A\neq A'}  p_{A'}^i \underbrace{\text{EV}_{\min}\left(T_A \otimes \OI_{\bar{A}} \left[\rho_{A'}^i \otimes \rho_{\bar{A'}}^i\right]\right)}_{\geq -\frac{1}{2}}\right]\\
& \qquad \qquad+ \underbrace{\text{EV}_{\min}\left(\I\right)}_{=1}\\
\geq& 0;
\end{align*}
since \(\sum_{A \neq A'}\) has two terms for 3 parties and \(\sum_{A'} \sum_i p^i_{A'}=1\). {This proves the first claim. To see that $c=1$ is optimal, it is enough to apply $\Phi_{T}$ to the bi-separable state on $ABC$ formed by a tensor product of a maximally entangled state of $AB$ and an arbitrary pure state of $C$.}
\end{proof}

The generalization of the above map to $n$ parties follows straightforwardly.

\begin{corollary}
\label{corphiTn}
For $\Phi_{T}= \sum_A T_A \otimes \OI_{\bar{A}} + \frac{2^{n-1}-2}{2}\, \I\cdot \text{{\it Tr}}$ and for any n-partite biseparable states $\rho_{2-sep}$ we have
\begin{align}
\Phi_T[{\rho}_{2-sep}] \geq 0.
\end{align}
\end{corollary}

\begin{proof}
As before we obtain
\begin{align*}
\begin{split}
\text{EV}_{\min}&(\Phi_T[\rho_{\text{2-sep}}])\\
&\geq \sum_{A'}\sum_i \left[ \sum_{A=A'} p^i_{A'} \underbrace{\text{EV}_{\min}\left(T_{A'} \otimes \OI_{\bar{A'}} \left[ \rho_{A'}^i \otimes \rho_{\bar{A'}}^i\right]\right)}_{\geq 0}+\sum_{A\neq A'} p^i_{A'} \underbrace{\text{EV}_{\min}\left(T_A \otimes \OI_{\bar{A}} \left[ \rho_{A'}^i \otimes \rho_{\bar{A'}}^i\right]\right)}_{\geq -\frac{1}{2}}\right]\\
& \qquad \qquad+ \frac{2^{n-1}-2}{2} \underbrace{\text{EV}_{\min}\left(\I\right)}_{=1}\geq 0;
\end{split}
\end{align*}
since this time \(\sum_{A \neq A'}\) has \(2^{n-1}-2\) terms.
\end{proof}

Theorem \ref{thmphiTn} of the main text follows from Corollary \ref{corphiTn} since \(\text{EV}_{\min}\) is invariant under unitary transformations. Note however that already for 4 parties \(\Phi_T\) fails to detect the \(W\)-state, whereas \(\Phi_{Tx}\) detects the GHZ-state for any number of parties.

\subsection{Reduction map} \label{min eig R}

The reduction map is given by $\Lambda=\frac{1}{d-1}(\mathds{1}\Tr - I)$. Here we normalised it to make it trace-preserving.

\begin{lemma}
For the reduction map we have $\mu(\Lambda)=\frac{1}{d}$.
\end{lemma} 

\begin{proof}
On the one hand, the maximally entangled state $\ket{\varepsilon}= \sum_i \frac{1}{\sqrt{d}} \ket{ii}$ yields $\braket{\varepsilon|(\Lambda\otimes I)[\ket{\varepsilon}\!\!\bra{\varepsilon}]|\varepsilon}=-\frac{1}{d}$, showing that \(\mu(\Lambda) \geq \frac{1}{d}\). In fact,
\begin{align}
\begin{split}
\Lambda \otimes I [\ket{\varepsilon} \bra{\varepsilon}]\, &=\, \frac{1}{d(d-1)}\, \sum_{ij} ( \mathds{1} \Tr \otimes I - I ) [\ket{ii} \bra{jj}]\,\\
&=\, \frac{1}{d(d-1)}\, \sum_{ij} \left( \delta_{ij} \mathds{1} \otimes \ket{i} \bra{j} - \ket{ii} \bra{jj}\right)\, \\
&=\, \frac{1}{d-1} \left(\frac{ \mathds{1}}{d} - \ket{\varepsilon} \bra{\varepsilon}\right),
\end{split}
\end{align}
and so
\begin{equation}
\bra{\varepsilon} \Lambda \otimes I [\ket{\varepsilon} \bra{\varepsilon}] \ket{\varepsilon} = \frac{1}{d-1} \left( \frac{1}{d} -1 \right) = -\frac{1}{d}.
\end{equation}
On the other hand, with the notation of eq. (1) in \cite{LH}, we have $\Lambda\otimes I +\frac{1}{d}\mathds{1}\Tr=\frac{1}{d}\Phi[\frac{d}{d-1}, 0, -\frac{1}{n-1}]$, which is a positive map thanks to the characterisation given in Theorem 3 of \cite{LH}. Hence as \(\mu(\Lambda) = \min \{c \mid \Lambda \otimes I + c \mathds{1} \Tr \geq 0\}\), we have \(\mu(\Lambda) \leq \frac{1}{d}\).
\end{proof}

As a consequence we obtain the following.

\begin{proof}[Proof of Theorem \ref{thmreduction}]
The argument goes in complete analogy with the one developed for the proof of Theorem \ref{thmphiT} (see Appendix \ref{min eig T}). The fact that the coefficient $c$ is twice the minimal output eigenvalue comes from the fact that we are dealing with tripartite systems.
\end{proof}

\subsection{Breuer-Hall map}  \label{min eig BH}

The Breuer-Hall map is defined in any even dimension $d\geq 4$ as
\begin{equation}
\mathcal{B}\, =\, \frac{\mathds{1}\Tr - I - \mathcal{V}T}{d-2}\, ,
\end{equation}
where $\mathcal{V}(\cdot)\equiv V(\cdot)V^{\dag}$ with $V$ a unitary, skew-symmetric matrix, $V^{T}=-V$. In order to find $\mu(\mathcal{B})$ we have to work a bit harder.

\begin{lemma}
The Breuer-Hall map $\mathcal{B}$ satisfies $\mu(\mathcal{B})=\frac{1}{d}$.
\end{lemma}

\begin{proof}
Using the maximally entangled state as an ansatz gives once again $\braket{\varepsilon|(\Lambda\otimes I)[\ket{\varepsilon}\!\!\bra{\varepsilon}]|\varepsilon}=-\frac{1}{d}$. In fact,
\begin{equation}
\begin{split}
\Lambda \otimes I [\ket{\varepsilon} \bra{\varepsilon}]\, &=\, \frac{1}{d(d-2)} \, \sum_{ij} ( \mathds{1} \Tr \otimes I - I - \mathcal{V} T \otimes I)[ \ket{ii} \bra{jj}]\\
&=\, \frac{1}{d(d-2)} \, \sum_{ij} \left(\delta_{ij} \mathds{1} \otimes \ket{i} \bra{j} - \ket{ii} \bra{jj} - V\ket{j} \bra{i} V^\dagger \otimes \ket{i}\bra{j}\right)\\
&=\, \frac{1}{d-2}\, \bigg(\frac{ \mathds{1}}{d} - \ket{\varepsilon} \bra{\varepsilon} - \frac{1}{d} \sum_{ij} V\ket{j} \bra{i} V^\dagger \otimes \ket{i} \bra{j}\bigg)\, .
\end{split}
\end{equation}
Using \(V^T = -V\), which becomes  \(\bra{j} V \ket{i} = - \bra{i} V \ket{j}\) at the level of matrix elements, we obtain
\begin{equation}
\begin{split}
\bra{\varepsilon} \Lambda \otimes I (\ket{\varepsilon} \bra{\varepsilon}) \ket{\varepsilon} &= \frac{1}{d-2}\, \bigg( \frac{1}{d} -1 - \frac{1}{d^2} \sum_{ijkl} \bra{kk} ( V \ket{j} \bra{i} V^\dagger \otimes \ket{i} \bra{j}) \ket{ll}\bigg)\\
&= \frac{1}{d-2}\, \bigg( \frac{1-d}{d} - \frac{1}{d^2} \sum_{ijkl} \bra{k} V \ket{j} \delta_{ki} \bra{i} V^\dagger \ket{l} \delta_{jl}) \bigg)\\
&= \frac{1}{d-2}\, \bigg(\frac{1-d}{d} - \frac{1}{d^2} \sum_{ij} \bra{i} V \ket{j} \bra{i} V^\dagger \ket{j}\bigg)\\
&= \frac{1}{d-2}\, \bigg(\frac{1-d}{d} + \frac{1}{d^2} \sum_{ij} \bra{j} V \ket{i} \bra{i} V^\dagger \ket{j}\bigg)\\
&= \frac{1}{d-2}\, \bigg(\frac{1-d}{d} + \frac{1}{d^2} \sum_{j} \bra{j} \underbrace{V V^\dagger}_{= \mathds{1}} \ket{j}\bigg)\\
&= \frac{1}{d-2}\, \bigg( \frac{1-d}{d}+\frac{1}{d}\bigg)=-\frac{1}{d}.
\end{split}
\end{equation}
Now, we have to show that $\frac{1}{d-2}(\mathds{1}\Tr - I - \mathcal{V} T)\otimes I + \frac{1}{d}\mathds{1}\Tr$ is a positive map. That is, taking an arbitrary input $\ket{\psi}$, we must prove that
\begin{equation}
\mathds{1}\otimes \rho_{\psi} - \ket{\psi}\!\!\bra{\psi} - V_{A} \ket{\psi}\!\!\bra{\psi}^{T_{A}} V_{A}^{\dag} + \frac{d-2}{d}\,\mathds{1}\, \geq\, 0\, ,
\end{equation}
where $\rho_{\psi}=\Tr_{A}\ket{\psi}\!\!\bra{\psi}$ is the reduced state of $\psi$. We need a couple of preliminary facts:
\begin{itemize}
\item if $Q>0$ then $Q-\ket{\psi}\!\!\bra{\psi}\geq 0$ iff $\braket{\psi|Q^{-1}|\psi}\leq 1$;
\item $\mathds{1}\Tr - \mathcal{V}T$ is a completely positive map.
\end{itemize}
To show the first claim, it is enough to apply the hermitian invertible operator $Q^{-1/2}$ on both sides of the inequality, and to observe that for an unnormalised vector $\ket{\alpha}$ we have that $\mathds{1}-\ketbra{\alpha}{\alpha}$ iff $\braket{\alpha|\alpha}\leq 1$. The second claim follows from the identity $\I \Tr - \mathcal{V}T=\mathcal{V}(\I\Tr - T)$  and from the fact that the Choi state of $\I \Tr - T$ is proportional to the projector onto the antisymmetric subspace (i.e. it is an extremal Werner state).

As an easy consequence of the second claim, we deduce that
\begin{equation}
\mathds{1}\otimes \rho_{\psi} - V_{A} \ket{\psi}\!\!\bra{\psi}^{T_{A}} V_{A}^{\dag}\, =\, \left( (\mathds{1}\Tr - \mathcal{V}T)\otimes\OI\right)[\ketbra{\psi}{\psi}]\, \geq\, 0\, ,
\end{equation}
which in turn shows that $Q=\mathds{1}\otimes \rho_{\psi} - V_{A} \ket{\psi}\!\!\bra{\psi}^{T_{A}} V_{A}^{\dag} + \frac{d-2}{d}\I>0$.

Combining this with the first of the above claims, we have left to show that
\begin{equation}
\braket{\psi| Q^{-1}|\psi}\, =\, \bra{\psi} \left(\mathds{1}\otimes \rho_{\psi} - V_{A} \ket{\psi}\!\!\bra{\psi}^{T_{A}} V_{A}^{\dag} + \frac{d-2}{d}\,\mathds{1}\right)^{-1}\ket{\psi}\, \leq\, 1\, .
\end{equation}
First of all, rewrite the left-hand side by taking out of parenthesis the unitaries $V_{A}$. Now, we have to compute 
\begin{equation}
\left(\mathds{1}\otimes \rho_{\psi} - \ket{\psi}\!\!\bra{\psi}^{T_{A}} + \frac{d-2}{d}\,\mathds{1}\right)^{-1}\, \equiv\, R^{-1}
\end{equation}
and show that
\begin{equation}
\bra{\psi} (V_A R V_A^\dagger)^{-1} \ket{\psi} = \bra{\psi} (V_A^\dagger)^{-1} R^{-1} V_A^{-1} \ket{\psi}=\bra{\psi} V_A R^{-1} V_A^\dagger \ket{\psi} \leq 1.
\end{equation}
To express \(R^{-1}\), start by writing the Schmidt decomposition $\ket{\psi}=\sum_{i} c_i\ket{i i}$. We then note that $R$ preserves the subspaces $W_{i}\equiv\text{Span}\{\ket{i^{*}i}\}$ ($i=1,\ldots,d$) and $W_{ij}\equiv\text{Span}\{\ket{i^{*}j},\ket{j^{*}i}\}$ (with $1\leq i<j\leq d$), where the complex conjugation is taken in the same basis as the partial transposition. In fact, all the three addends appearing in the expression of $R$ preserve those subspaces, as follows easily from the argument in the proof of Lemma \ref{mu T}.

To describe $R$, we hence only need to describe its restriction to these later subspaces. Leveraging the results obtained in the proof of Lemma \ref{mu T}, we have
\begin{equation}
\begin{split}
R\big|_{W_i}&= 1-\frac{2}{d}\\
R\big|_{W_{ij}} &= \begin{pmatrix} c_{j}^2+1-\frac2d & -c_{i}c_{j} \\[1ex] -c_{i}c_{j} & c_{i}^2+1-\frac2d \end{pmatrix},
\end{split}
\end{equation}
which leads to
\begin{equation}
\begin{split}
R\big|_{W_i}^{-1}&= \frac{1}{1-\frac{2}{d}}\\
R\big|_{W_{ij}}^{-1}& = \frac{1}{\left(1-\frac2d\right)\left(c_{i}^2+c_{j}^2+1-\frac2d\right)}\, \begin{pmatrix} c_{i}^2+1-\frac2d & c_{i}c_{j} \\[1ex] c_{i}c_{j} & c_{j}^2+1-\frac2d \end{pmatrix}.
\end{split}
\end{equation}
Putting all together, we see that
\begin{align*}
R^{-1}\, &=\, \frac{1}{1-\frac2d}\, \sum_{i} \ket{i^*i}\!\!\bra{i^*i}\, +\, \frac{1}{1-\frac2d}\, \sum_{i<j}\, \frac{1}{c^2_{i}+c^2_{j}+1-\frac2d} \, \Bigg( \left(c^2_{i}+1-\frac2d\right) \ket{i^*j}\!\!\bra{i^*j}\, +\\
&+\, c_{i}c_{j}\, \left(\ket{i^*j}\!\!\bra{j^* i}\, +\, \ket{j^*i}\!\!\bra{i^* j}\right)\, +\, \left( c_{j}^2+1-\frac2d\right) \ket{j^*i}\!\!\bra{j^* i} \Bigg)\, .
\end{align*}
Now that we have an expression for \(R^{-1}\), $\braket{\psi|V_{A}R^{-1}V_{A}^{\dag}|\psi}\leq 1$ is left to prove, the matrix elements of which can be easily computed with the help of the identities
\begin{align}
&\braket{j|V|i^*}_A\braket{i^*|V^{\dag}|j}_A\, =\, \braket{i|V|j^*}_A\braket{j^*|V^{\dag}|i}_A\, =\, \big|\braket{i|V|j^*}_A\big|^{2}\, , \\
&\braket{j|V|i^*}_A\braket{j^*|V^{\dag}|i}_A\, =\, \braket{i|V|j^*}_A\braket{i^*|V^{\dag}|j}_A\, =\, -\,\big|\braket{i|V|j^*}_A\big|^{2}\, , \\
&\braket{{i}|V|{i}^{*}}_A\, =\, \braket{{i}^{*}|V|{i}}_A\, =\, 0\, ,
\end{align}
all consequences of the fundamental fact that $V$ is skew-symmetric, i.e. $V^{T}=-V$. Using those identities we get the following.
\begin{equation}
\begin{split}
\braket{\psi|V_{A}R^{-1}V_{A}^{\dag}|\psi}& = \sum_{kl} c_k c_l \bra{k k} V_A R^{-1} V_A^\dagger \ket{l l}\\
&= \sum_{kl} c_k c_l \left[ \frac{1}{1-\frac{2}{d}} \sum_i \underbrace{\bra{k k} V_A \ket{i^* i}}_{=\delta_{ki} \bra{i}V\ket{i^*}_A} \underbrace{ \bra{i^* i}V_A^{\dagger} \ket{l l}}_{= \delta_{il} \bra{i^*} V\ket{i}_A}+ \right.\\
&\qquad +\left. \frac{1}{1- \frac{2}{d}} \sum_{i<j} \frac{1}{c^2_i+c^2_j+1-\frac{2}{d}}\left((c^2_i+1-\frac{2}{d}) \underbrace{ \bra{k k} V_A \ket{i^* j}}_{= \delta_{kj} \bra{j} V\ket{i^*}_A} \underbrace{\bra{i^* j} V_A^\dagger \ket{l l}}_{= \delta_{jl} \bra{i^*} V\ket{j}_A}+\right. \right.\\
&\qquad \qquad \qquad +\left.\left. c_i c_j \underbrace{\bra{k k} V_A \ket{i^* j}}_{= \delta_{kj} \bra{j} V\ket{i^*}_A} \underbrace{ \bra{j^* i}V_A^\dagger \ket{l l}}_{= \delta_{il} \bra{j^*} V\ket{i}_A}+c_i c_j \underbrace{\bra{k k} V_A \ket{j^* i}}_{= \delta_{ki} \bra{i} V\ket{j^*}_A} \underbrace{\bra{i^* j}V_A^\dagger \ket{l l}}_{= \delta_{jl} \bra{i^*} V\ket{j}_A} \right. \right.\\
&\left. \left. \qquad \qquad \qquad (c^2_j+1-\frac{2}{d})  \underbrace{\bra{k k} V_A \ket{j^* i}}_{= \delta_{ki} \bra{i} V\ket{j^*}_A} \underbrace{\bra{j^* i}V_A^\dagger \ket{l l}}_{= \delta_{il} \bra{j^*} V\ket{i}_A}\right)\right] \\
&= \frac{1}{1-\frac{2}{d}} \sum_{i<j} \frac{1}{c^2_i+c^2_j +1-\frac{2}{d}} \left( \big(c^2_i+1-\frac{2}{d}\big) c^2_j \lvert \bra{i} V \ket{j^*}_A\rvert^2 - 2 c^2_i c^2_j \lvert \bra{i} V \ket{j^*}_A\rvert^2 + \right. \\
&\qquad \left.+ \big(c^2_j+1-\frac{2}{d}\big) c^2_i \lvert \bra{i} V \ket{j^*}_A\rvert^2 \right)\\
 &=\sum_{i<j}\, \frac{c^2_{i}+c^2_{j}}{c^2_{i}+c^2_{j}+1-\frac2n}\, \big|\braket{i|V|j^*}_A\big|^{2}\, =\, \frac12\, \sum_{i, j}\, \frac{c^2_{i}+c^2_{j}}{c^2_{i}+c^2_{j}+1-\frac2n}\, \big|\braket{i|V|j^*}_A\big|^{2}\, ,
\end{split}
\end{equation}
where the last step is possible thanks to $\big|\braket{i|V|j^*}_A\big|^{2}= \bra{j} V \ket{i^*}_A \bra{i^*} V^\dagger \ket{j}_A=\big|\braket{j|V|i^*}_A\big|^{2}$, which is moreover zero if $i=j$. We have to prove that the right-hand side of the above equation is always upper bounded by $1$, for all orthonormal basis $\{\ket{i}_A\}_{i}$ and all probability distributions $\{c^2_{i}\}_{i}$. Observe that the matrix $S$ defined via $S_{ij}=\big|\braket{i|V|j^*}_A\big|^{2}$ is clearly doubly stochastic, as \(\sum_j S_{ij}= \sum_j \lvert \bra{i} V \ket{j^*}_A \rvert^2 = \sum_j \bra{i} V \ket{j^*}_A\bra{j^*} V^\dagger \ket{i}_A= \bra{i} V V^\dagger\ket{i}_A=1\) and similarly \(\sum_i S_{ij}=1\). Hence by Birkhoff's Theorem it is the convex combination of permutation matrices, and by linearity we can limit ourselves to show that
\begin{equation} \label{equ:tomaxfunc}
\frac12\, \sum_{i}\, \frac{c^2_{i}+c^2_{\sigma(i)}}{c^2_{i}+c^2_{\sigma(i)}+1-\frac2d}\, \leq\, 1\, .
\end{equation}
for all permutations $\sigma$. Defining $p_{i}\equiv \frac{c^2_{i}+c^2_{\sigma(i)}}{2}$, we rephrase Equation \ref{equ:tomaxfunc} as
\begin{equation}\label{equ:fp}
\sum_{i}\, \frac{p_{i}}{2p_{i}+1-\frac2d}\, \leq\, 1\, ,
\end{equation}
to be proven for all probability distributions $p$. The above function of $p$ is easily seen to be concave and achieving its maximum for $p_{i}\equiv \frac1d$, which yields the required upper bound.\end{proof}

\section{Proof of Proposition \ref{prop-main}}\label{Aprojection}
In the following we prove Proposition \ref{prop-main} of the main text. To do this, we first prove the following

\begin{proposition}
\label{propGHZmat}
For all \(n \in \mathbb{N},\; n \geq 2\) we have for the mixture of local unitaries \(\mathcal{X}_{n}^d\) that 
\begin{align}\label{eq.proj}
\mathcal{X}_{n}^d ( (\ket{+} \bra{+})^{\otimes n})=\text{GHZ}_{n,\text{cyclic}}^{\;d} ,
\end{align}
where \(\ket{+} \coloneqq \ket{0}+ \ket{1} + \dots+ \ket{d-1}\) and  \((\ket{+} \bra{+})^{\otimes n}\) denotes the matrix whose elements are all equal to one.
\end{proposition}

\begin{proof}
For \(n=2\) we have:
\begin{align*}
\mathcal{X}^d_{2}=&h_{2,2}^d((\ket{+} \bra{+})^{\otimes 2})\\
=&\frac{1}{d} \sum_{k=0}^{d-1} Z_k^1 (Z_k^2)^{\dagger} \;(\ket{+}\bra{+})^{\otimes 2} \;(Z_k^1)^{\dagger} Z_k^2 \\
 \stackrel{*}{=}&\quad \;( \ket{00} + \ket{11}+\dots + \ket{(d-1)(d-1)}) (\bra{00} + \bra{11} + \dots+\bra{(d-1) (d-1)})\\
 &+( \ket{01} + \ket{12}+\dots+\ket{(d-1) 0}) (\bra{01} + \bra{12} +\dots+ \bra{(d-1) 0})\\
 &+ \dots \\
 &+( \ket{0(d-1)} + \ket{10}+\dots+ \ket{(d-1)(d-2)}) (\bra{0(d-1)} + \bra{10} + \bra{(d-1)(d-2)})\\
 =& \text{GHZ}_{2,\text{cyclic}}^d.
\end{align*}

To see that \(\stackrel{*}{=}\) holds one can, representing \(Z_k \) by \(\text{diag} (1,(k), (2k), \dots, ((d-1)k))\), where for \(l \in \{0, \dots, (d-1)\},\quad (lk)\coloneqq e^{2 \pi i l k/d}\), view \(Z_k \otimes (Z_k)^{\dagger} (\ket{+}\bra{+})^{\otimes n} (Z_k)^{\dagger} \otimes Z_k\) as

\[
\begin{array}{c|ccccc|ccccc|ccc|}
&1& (k)&(2k)& \cdots & ((d-1) k)& (-k)& 1&\cdots & ((d-2) k)& \cdots&((-d+1) k) &\cdots &1\\
\hline

1&1& (k)&(2k)& \cdots & ((d-1) k)& (-k)& 1&\cdots & ((d-2) k)& \cdots&((-d+1) k) &\cdots &1\\
(-k)& & 1& & \cdots & &  & &\cdots & & \cdots& 1&\cdots & \\
(-2k)& & &1 & \cdots & &  & &\cdots & & \cdots&  &\cdots &\\
\vdots& & & & \ddots & &  & &\ddots & & \ddots&  &\ddots &\\
((-d+1) k)& & & & \cdots & 1& 1 & &\cdots & & \cdots&  &\cdots &\\
\hline
(k)& & & & \cdots & 1& 1 & &\cdots & & \cdots&  &\cdots &\\
1&1& (k)&(2k)& \cdots & ((d-1) k)& (-k)& 1&\cdots & ((d-2) k)& \cdots&((-d+1) k) &\cdots &1\\
(-k)& & 1& & \cdots & &  & &\cdots & & \cdots& 1&\cdots & \\
\vdots& & & & \ddots & &  & &\ddots & & \ddots&  &\ddots &\\
((-d+2)k)& & & & \cdots & &  & &\cdots & 1& \cdots&  &\cdots &\\
\hline

\vdots& & & & \ddots & & & &\ddots & & \ddots& &\ddots &\\
\hline
((d-1) k)& & 1& & \cdots & &  & &\cdots & & \cdots& 1&\cdots & \\
((d-2) k)& & &1 & \cdots & &  & &\cdots & & \cdots&  &\cdots &\\

\vdots& & & & \ddots & &  & &\ddots & & \ddots& &\ddots &\\
1&1& (k)&(2k)& \cdots & ((d-1) k)& (-k)& 1&\cdots & ((d-2) k)& \cdots&((-d+1) k) &\cdots &1\\

\end{array}.
\]
The leftmost column (left of the vertical bar) represents \(Z_k \otimes (Z_k)^{\dagger}\) as the latter multiplies each row of \((\ket{+} \bra{+})^{\otimes n}\) by the displayed factor. Similarly, \((Z_k)^{\dagger} \otimes Z_k\) is represented by the topmost row since it multiplies each column of \((\ket{+} \bra{+})^{\otimes n}\) by the displayed factor.

Now, we see that the first row (\(\ket{00}\)) has ones exactly at \(\bra{00}, \bra{11}, \dots, \bra{(d-1) (d-1)}\) such that the 
first row of this matrix can be represented by \(\ket{00} \bra{\text{GHZ}}\). The second row where a 1 appears on the leftmost
column is the \(\ket{11}\) row, and again this row can be represented by \(\ket{11} \bra{\text{GHZ}}\). Hence we see that we can
represent the matrix where all the ones of the topmost row and leftmost column meet by \(\ket{\text{GHZ}} \bra{\text{GHZ}}\). 
Similarly, the matrix representing where all the (k) of the topmost column and (-k) of the leftmost row meet can be written 
as \(X_d^{(1)} \ket{\text{GHZ}} \bra{\text{GHZ}} (X_d^{(1)})^{\dagger}\). In general, for \(l \in \{0, 1, \dots, d-1\}\), the matrix 
representing where all the \((lk)\) of the topmost column and \((- lk)\) of the leftmost row meet can be written
as \(X_d^{(l)} \ket{\text{GHZ}} \bra{\text{GHZ}} (X_d^{(l)})^{\dagger}\). Those matrices represent all the entries 
where \(Z_k \otimes (Z_k)^{\dagger} (\ket{+}\bra{+})^{\otimes n} (Z_k)^{\dagger} \otimes Z_k\) has ones. The other entries 
of this matrix have a phase \((lk)\) for some \(l \in \{0, 1, \dots, d-1\}\). When summing over all the \(k = 0, \dots, d-1\) we have
\[\sum_{k=0}^{d-1} (lk) = \sum_{k=0}^{d-1} e^{2\pi i lk/d} = \frac{1-e^{2\pi i d/d}} {1-e^{2\pi i l/d}}=0;\]
such that we have
\[\sum_{k=0}^{d-1}Z_k \otimes (Z_k)^{\dagger} (\ket{+}\bra{+})^{\otimes n} (Z_k)^{\dagger} \otimes Z_k= \sum_{k=0}^{d-1} \sum_{l=0}^{d-1}X_d^{(l)} \ket{\text{GHZ}} \bra{\text{GHZ}} (X_d^{(l)})^{\dagger}= d \cdot \text{GHZ}_{2, \text{cyclic}}^d.\]
  This proves the Claim for \(n=2\). For \(n > 2\) notice that
\begin{align*} 
h_{n,2}^d((\ket{+} \bra{+})^{\otimes n})=&\frac{1}{d} \sum_{k=0}^{d-1} Z_k^1 (Z_k^2)^{\dagger} \;(\ket{+}\bra{+})^{\otimes n} \;(Z_k^1)^{\dagger} Z_k^2\\
=& \frac{1}{d} \sum_{k=0}^{d-1} \De{Z_k \otimes (Z_k)^{\dagger} \;(\ket{+}\bra{+})^{\otimes 2} \;(Z_k)^{\dagger}\otimes Z_k }\; \otimes (\ket{+} \bra{+})^{\otimes (n-2)}\\
 =& \sum_{l=0}^{d-1}X_d^{(l)} \ket{\text{GHZ}} \bra{\text{GHZ}} (X_d^{(l)})^{\dagger} \otimes \sum_{\substack{i_3, \dots, i_n=0\\ j_3, \dots, j_n =0}}^{d-1} \ket{i_3 \dots i_n} \bra{j_3 \dots j_n} \\
  =&\sum_{\substack{i_3, \dots, i_n=0\\ j_3, \dots, j_n =0}}^{d-1} \sum_{C_2 \in \{0,\dots,d-1\}} X_d^{(C_2,0,\dots,0)} \ket{\text{GHZ}} \otimes  \ket{i_3 \dots i_n} \bra{\text{GHZ}} \otimes  \bra{j_3 \dots j_n} (X_d^{(C_2, 0,\dots,0)})^{\dagger}.
  \end{align*}

Assuming that we are working in the following basis 
of \(\mathbb{C}^d \otimes \dots \otimes \mathbb{C}^d\): \(( e_1^1 \otimes \dots \otimes e_n^1 , e_1^1 \otimes \dots \otimes e_n^2, \dots, e_1^1 \otimes \dots \otimes e_n^d, e_1^1 \otimes \dots \otimes e_{n-1}^2\otimes e_n^1, \dots, e_1^d \otimes \dots \otimes e_n^d) \) we define the basis 
transformation \(D_{2k}\) by
\[ D_{2k}(e_1^{l_1}\otimes e_2^{l_2} \otimes \dots \otimes e_k^{l_k} \otimes \dots \otimes e_n^{l_n}) \coloneqq e_1^{l_1}\otimes e_2^{l_k} \otimes \dots \otimes e_k^{l_2} \otimes \dots \otimes e_n^{l_n},\]
\(l_1, \dots, l_n \in \{1,\dots,d\}.\)
This has the effect of exchanging the subspaces 2 and k. Indeed we have 
\begin{align*}
h_{n,k}^d&= D_{2k} h_{n,2}^d D_{2k}^{-1},\\
X^{(C_2, \dots, C_k, \dots ,C_n)}_d&= D_{2k} X^{(C_k, \dots, C_2, \dots ,C_n)}_d D_{2k}^{-1};
\end{align*}
 and hence
 \begin{align*}
 h_{n,k}^d ((\ket{+} \bra{+})^{\otimes n}) =& D_{2k} h_{n,2}^d D_{2k}^{-1} ((\ket{+} \bra{+})^{\otimes n})\\
 =& D_{2k} h_{n,2}^d ((\ket{+} \bra{+})^{\otimes n}) D_{2k}^{-1} \\
 =& \sum_{\substack{i_3, \dots, i_n=0\\ j_3, \dots, j_n =0}}^{d-1} \sum_{C_2 \in \{0, \dots, d-1\}}  X_d^{(0,\dots,C_2, \dots, 0)} \ket{\text{GHZ}}_{1k} \otimes \ket{i_k, i_3, \dots, i_{k-1}, i_{k+1},\dots, i_n} \\
 &\qquad \qquad  \quad \bra{\text{GHZ}}_{1k} \otimes \bra{i_k, i_3, \dots, i_{k-1}, i_{k+1},\dots, i_n}  ( X_d^{(0,\dots,C_2, \dots, 0)} )^{\dagger},
 \end{align*}
where \(\ket{\text{GHZ}}_{1k} \otimes \ket{i_k, i_3, \dots, i_{k-1}, i_{k+1}, \dots, i_n} \coloneqq \sum_{l=0}^{d-1} \ket{l, i_k,i_3, \dots , i_{k-1}, l, i_{k+1}, \dots, i_n}\).
Therefore, since \(X_d^{(C_2,0,\dots,0)} X_d^{(0,C_3,0,\dots,0)} \dots X_d^{(0,\dots,0, C_n)} = X_d^{\mathbf{C}}\) we find
\begin{align*}
\mathcal{X}_{n}^{d}((\ket{+} \bra{+})^{\otimes n})&= h_{n,2}^d \circ h_{n,3}^d \circ \dots \circ h_{n,n}^d[(\ket{+} \bra{+})^{\otimes n}]\\
&=\sum_{\mathbf{C} \in \{0,\dots,d-1\}^{n-1}} X_d^{\mathbf{C}} \ket{\text{GHZ}} \bra{\text{GHZ}} (X_d^{\mathbf{C}})^{\dagger}\\
&= \text{GHZ}_{n,\text{cyclic}}^d. 
\end{align*}
\end{proof}
With this we can now easily prove Proposition \ref{prop-main} of the main text.

\begin{proof}[Proof of Proposition \ref{prop-main}]
By noting that for any $n$-qudit state \(\rho\) we have \(\mathcal{X}_n^d (\rho) = \mathcal{X}_n^d(\ket{+} \bra{+}^{\otimes n}) \circ_s \rho\), where \(\circ_s\) denotes the Schur product, Proposition \ref{propGHZmat} means exactly that \(\mathcal{X}_n^d\) projects any $n$-qudit state into the subspace spanned by the entries of the \textit{GHZ-like cyclically permuted} states, which is the statement of Proposition \ref{prop-main} of the main text.
\end{proof}

Although this way of proving Proposition \ref{prop-main} gives great insights into how the projection is concretely performed, we would 
like to present an alternative way of proving it that, using a group theoretical approach, puts more emphasis on the symmetries of the GHZ 
cyclically permuted states that are made used of here, i.e. the symmetry of the Choi map. 

\begin{proof}[Alternative Proof of Proposition \ref{prop-main}]
Let $\mathcal{G}$ be a compact group whose associated Haar integral we denote by $\int_\mathcal{G}$. Given a unitary representation $\zeta:\mathcal{G}\rightarrow \mathcal{L}(\mathcal{H})$ of $\mathcal{G}$ on a finite--dimensional complex Hilbert space, an interesting super-operator is given by the expression
\begin{equation}
\mathcal{P}\ \equiv\ \int_\mathcal{G} dg\ \text{Ad}_{\zeta}(g) :\ \mathcal{L}(\mathcal{H})\longrightarrow\mathcal{L}(\mathcal{H})\ , \label{P}
\end{equation}
where $\text{Ad}_{\zeta}(g)$ acts on linear operators $X\in\mathcal{L}(\mathcal{H})$ as  $\left(\text{Ad}_{\zeta}(g)\right)(X)\equiv\zeta(g)X\zeta(g)^\dag$. The super-operator $\mathcal{P}$ given by \eqref{P} is an orthogonal projector with respect to the Hilbert--Schmidt product on $\mathcal{L}(\mathcal{H})$. Moreover, let the irreducible decomposition of $\mathcal{L}(\mathcal{H})$ under the action of $\zeta$ be given by
\begin{align}
\mathcal{H}\ &=\ \bigoplus_\alpha V_\alpha^{\oplus n_\alpha}\ =\ \bigoplus_\alpha V_\alpha\otimes \mathds{C}^{n_\alpha}\ , \label{irr L} \\
\zeta\ &=\ \sum_\alpha n_\alpha \zeta_\alpha\ , \label{irr zeta}
\end{align}
with $\zeta_\alpha$ irreducible for all $\alpha$ and $\zeta_\alpha,\zeta_{\alpha^\prime}$ inequivalent if $\alpha\ne\alpha'$. Then $\mathcal{P}$ acts as
\begin{equation}
\mathcal{P}(\cdot)\ =\ \bigoplus_\alpha\, \frac{\mathds{1}}{d_\alpha}\otimes \Tr_{V_\alpha}[\Pi_\alpha\, (\cdot)\, \Pi_\alpha]\ , \label{eq lemma Schur}
\end{equation}
where $d_\alpha\equiv\dim V_\alpha$, $\Pi_\alpha$ is the projector onto the $\alpha$--th block of the direct sum in \eqref{irr L}, and the partial trace is over the first component of the bipartite Hilbert space $V_\alpha\otimes \mathds{C}^{n_\alpha}$.

Despite the complicated appearance, the above observation is just a slightly more sophisticated application of Schur's Lemma.

In our case, we choose $\mathcal{H}=(\mathds{C}^d)^{\otimes n}$ and $\mathcal{G}=\mathds{Z}_d^{n-1}$, where $\mathds{Z}_d$ is the group of integers with the operation of sum modulo $d$. The Haar integral over $\mathcal{G}$ is simply the normalized sum, i.e. $\int_\mathcal{G}=\frac{1}{d^{n-1}}\sum_{g\in \mathcal{G}}$. For $g=(k_2,\ldots,k_n)\in \mathcal{G}$, we define 
\begin{equation} \zeta(g)\,\equiv\, Z_{k_2+\ldots +k_n}\otimes Z^\dag_{k_2} \otimes\ldots\otimes Z^\dag_{k_n}\, . \end{equation}
Then, it is not difficult to see that $\chi_n^d = \mathcal{P}$ as defined by Eq. \ref{eqmapproj} of section \ref{sec:choi} of the main text and \eqref{P} here.
Now, all we have to do is to decompose $\zeta$ in irreducible representations (irreps), all one-dimensional because $\mathcal{G}$ is abelian. For all $0\leq r, h_2,\ldots,h_n\leq d$, we observe that $\ket{r,r+h_2,\ldots, r+h_n}$ satisfies
\begin{equation}
\zeta(k_2,\ldots,k_n) \ket{r,r+h_2,\ldots, r+h_n} = \omega^{-h_2 k_2 -\ldots -h_n k_n} \ket{r,r+h_2,\ldots, r+h_n}\, ,
\end{equation}
i.e. it is an irreducible subspace for $\zeta$. It turns out that the irreps of $\mathcal{G}$ are all present in $\zeta$ and indexed by $h_2,\ldots, h_n$, and each of them has multiplicity $d$ (internal index $r$). This can be easily seen by verifying that the characters of $\ket{r,r+h_2,\ldots, r+h_n}$ and $\ket{r',r'+h'_2,\ldots, r'+h'_n}$ are orthogonal if $(h_2,\ldots,h_n)\neq (h'_2,\ldots ,h'_n)$, and by remembering that the number of non-isomorphic irreps of an abelian group coincides with its cardinality ($d^{n-1}$ in this case). Finally, the action of $\mathcal{P}$ as specified in \eqref{eq lemma Schur} gives exactly the claim of Proposition \ref{prop-main}.
\end{proof}

\section{Choi Map}\label{Aprooflemma}
We want here to prove Lemma \ref{lemmaODT} and Lemma \ref{lemmaOD} of the main text. We begin by proving Lemma \ref{lemmaOD} that we state
again for ease of read.

\begin{lemma}[Lemma \ref{lemmaOD}]
 \begin{align}
\text{OD}(\phi[\mathcal{X}_{n}^d[\rho]]) = (2^{n-1}-1)\text{OD}[\Lambda_A\otimes I_{\bar{A}}[\mathcal{X}_{n}^d[\rho]]]\;\; \forall \, A,
 \end{align}
 where $\text{OD}[X]=X-\text{Diag}[X]$.
\end{lemma}

\begin{proof}
First of all note that
\begin{align*}
\text{GHZ}_{n,\text{cyclic}}^d &= \sum_{C \in \{0,1,\dots, d-1\}^{n-1}} X_d^C \ket{\text{GHZ}_n^d} \bra{\text{GHZ}_n^d} (X_d^C)^{\dagger}\\
&= \sum_{C \in \{0,1,\dots, d-1\}^{n-1}} \frac{1}{d} \sum_{i,j=0}^{d-1} X_d^C \ket{i\, i \dots i} \bra{j \, j \dots j} (X_d^C)^{\dagger}\\
&= \sum_{C \in \{0,1,\dots, d-1\}^{n-1}} \frac{1}{d} \sum_{i,j=0}^{d-1}  \ket{i \; i+C_2 \dots i+C_n} \bra{j \; j+C_2 \dots j+C_n},
\end{align*}
where the sums \(k+C_l, \quad k=i,j , \; l = 2, \dots, n\) are to be understood as modulo d. For any state \(\rho= (\rho_{\vec{k_1},\vec{k_2}}) \geq 0,\) 
with \(\vec{k_1},\vec{k_2} \in \{0, \dots, d-1\}^n\) we hence have
\begin{align*}
\mathcal{X}_{n}^d[\rho]&= \mathcal{X}_{n}^d[\ket{+} \bra{+}^{\otimes n}] \circ_s \rho\\
&= \text{GHZ}_{n,\text{cyclic}}^d \circ_s \rho\\
&= \sum_{C \in \{0,1,\dots, d-1\}^{n-1}} \frac{1}{d} \sum_{i,j=0}^{d-1} \rho_{i+(0,C), j+(0,C)}  \ket{i \; i+C_2 \dots i+C_n} \bra{j \; j+C_2 \dots j+C_n}.
\end{align*}

Note that trivially \(i \neq j \Leftrightarrow i+C_k \neq j +C_k, \; \forall k \in \{2,\dots,n\}\) such that, the main diagonal excepted, 
all non-vanishing elements of \(\mathcal{X}_{n}^d(\rho)\) are off-diagonal with respect to any partition (\(A \mid \bar{A}\)). 
In particular, this means that for any partition (\(A \mid \bar{A}\))
 \begin{equation}
 \Lambda_A\otimes I_{\bar{A}} [\text{OD}[\mathcal{X}_{n}^d[\rho]]]= - \text{OD}[\mathcal{X}_{n}^d[\rho]];
 \label{equ:ODpartition}
 \end{equation}
 that is the diagonal part of \(\Lambda_A\) acts trivially.  In other words, the above means that \(\text{Diag}_A\) only affects the diagonal
 elements of \(\mathcal{X}_n^d[\rho]\). Furthermore, since the latter becomes diagonal upon applying \(\text{Diag}_A\), it follows
 \begin{equation}
 \Lambda_A \otimes I_{\bar{A}} [ \text{OD} [ \mathcal{X}_n^d [\rho]]] = \text{OD} [\Lambda_A \otimes I_{\bar{A}} [ \mathcal{X}_n^d [\rho]]],
 \label{equ:ODcommute}
 \end{equation}
 such that
 \begin{align*}
 \text{OD}[\phi [\mathcal{X}_n^d[\rho]]] &= \text{OD}\left[ \sum_{B \mid \bar{B}} \Lambda_B \otimes I_{\bar{B}}[\mathcal{X}_n^d[\rho]]\right]\\
 &= \sum_{B \mid \bar{B}} \text{OD} [ \Lambda_B \otimes I_{\bar{B}} [\mathcal{X}_n^d [\rho]]]\\
 &\myequ{\ref{equ:ODcommute}} \sum_{B \mid \bar{B}} \Lambda_B \otimes I_{\bar{B}} [ \text{OD} [\mathcal{X}_n^d[\rho]]]\\
 &\myequ{\ref{equ:ODpartition}} \sum_{B \mid \bar{B}} ( - \text{OD} [\mathcal{X}_n^d[\rho]])\\
 &=(2^{n-1}-1) (- \text{OD}[ \mathcal{X}_n^d[\rho]])\\
 &\myequ{\ref{equ:ODpartition}}(2^{n-1}-1) \Lambda_A \otimes I_{\bar{A}} [\text{OD} [ \mathcal{X}_n^d[\rho]]], \quad \forall A.
 \end{align*}
\end{proof}

Lemma \ref{lemmaODT} is the same assertion as Lemma \ref{lemmaOD} but for \(\tilde{\sigma}_x^A \circ T_A\) instead of \(\Lambda_A\). The proof is also identical and can be carried through by systematically replacing \(\Lambda_A\) by \(\tilde{\sigma}_x^A \circ T_A\), setting \(d=2\) and noting that \(\mathcal{X}_n^2 \equiv \mathcal{X}_n\) in the above proof. Only equation \ref{equ:ODpartition} has to be adapted to
\begin{equation}
\tilde{\sigma}_x^A\circ T_A \otimes \OI_{\bar{A}} [\text{OD} [\mathcal{X}_n[\rho]]] = \text{OD}[\mathcal{X}_n[\rho]]
\end{equation}
since \(\tilde{\sigma}_x^A \circ T_A \otimes \OI_{\bar{A}}\) acts as the identity on \(\text{OD} [ \mathcal{X}_n[\rho]]\), not as minus the identity.

\section{A 3-qutrit PPT across all cuts GME family\label{app.qutritppt}}

Here we define a class of 3-qutrit GHZ-like states that were first introduced in \cite{hubersengupta} and were found in that same article 
for some parameter range to be GME although being positive under partial transpose (PPT) across all cuts. Following the footsteps
of \cite{hubersengupta}, for a given subset \(A \subset \{1,2,3\}\) and numbers \(x,y \in \{0,1,2\}, \; \lambda \in \mathbb{R}_+\) we first define
\begin{equation}
\ket{A,x,y,\lambda}:= \prod_{i \in A} \sigma^{(x,y)}_{\{i\}} \otimes \I_{\bar{\{i\}}}( \sqrt{\lambda} \ket{xxx}+\sqrt{\lambda^{-1}} \ket{yyy}),
\end{equation} 
where the swap matrix \(\sigma^{(x,y)}\) is defined by \(\sigma^{(x,y)}:= \ket{x} \bra{y} + \ket{y} \bra{x}\). We denote the
density matrix associated to such a (unnormalized) state by
\begin{equation}
P_A(x,y,\lambda):= \ket{A,x,y,\lambda} \bra{A,x,y,\lambda}.
\end{equation} 

We then define
\begin{equation}
E(\lambda_1,\lambda_2,\lambda_3):= 3 \ket{\text{GHZ}_3} \bra{\text{GHZ}_3} + \sum_{i=1,2,3} \sum_{x<y} \sum_{y=1,2} P_{\{i\}}(x,y,\lambda_i),
\end{equation}
with \(\ket{\text{GHZ}_3}= \frac{1}{\sqrt{3}} (\ket{000}+\ket{111}+\ket{222})\), which finally gives us the normalized state
\begin{equation}
\rho(\lambda_1,\lambda_2,\lambda_3):= \frac{E(\lambda_1,\lambda_2,\lambda_3)}{\Tr(E(\lambda_1,\lambda_2,\lambda_3))}.
\end{equation}

For the main text, the states of interest are the ones where \(\lambda_1=\lambda_2=\lambda_3=\lambda\), which we denote by \(\rho(\lambda)\), i.e. \(\rho(\lambda):= \rho(\lambda,\lambda,\lambda)\).

Alternatively, note that as
\begin{equation}
\begin{split}
\ket{\{1\},x,y,\lambda_1}&= \sqrt{\lambda_1} \ket{yxx}+\sqrt{\lambda_1^{-1}} \ket{xyy}\\
\ket{\{2\},x,y,\lambda_2}&= \sqrt{\lambda_2} \ket{xyx}+\sqrt{\lambda_2^{-1}} \ket{yxy}\\
\ket{\{3\},x,y,\lambda_3}&= \sqrt{\lambda_3} \ket{xxy}+\sqrt{\lambda_3^{-1}} \ket{yyx},
\end{split}
\end{equation}
we get by defining
\begin{equation}
\begin{aligned}
&\ket{\psi_1}:= \ket{\{3\},0,1,a}=\sqrt{a} \ket{001}+\sqrt{a^{-1}} \ket{110}; &&\qquad \quad \ket{\psi_6}:= \ket{\{1\},1,2,b}=\sqrt{b} \ket{211}+\sqrt{b^{-1}} \ket{122}\\
&\ket{\psi_2}:= \ket{\{2\},0,1,a}=\sqrt{a} \ket{010}+\sqrt{a^{-1}} \ket{101}; &&\qquad \quad \ket{\psi_7}:= \ket{\{3\},0,2,c^{-1}}=\sqrt{c^{-1}} \ket{002}+\sqrt{c} \ket{220}\\
&\ket{\psi_3}:= \ket{\{1\},0,1,a}=\sqrt{a} \ket{100}+\sqrt{a^{-1}} \ket{011}; &&\qquad \quad \ket{\psi_8}:= \ket{\{2\},0,2,c^{-1}}=\sqrt{c^{-1}} \ket{020}+\sqrt{c} \ket{202}\\
&\ket{\psi_4}:= \ket{\{3\},1,2,b}=\sqrt{b} \ket{112}+\sqrt{b^{-1}} \ket{221}; &&\qquad \quad \ket{\psi_9}:= \ket{\{1\},0,2,c^{-1}}=\sqrt{c^{-1}} \ket{200}+\sqrt{c} \ket{022}\\
&\ket{\psi_5}:= \ket{\{2\},1,2,b}=\sqrt{b} \ket{121}+\sqrt{b^{-1}} \ket{212}; &&\qquad \quad \ket{\psi_{10}}:= \ket{000}+\ket{111}+\ket{222},\\
\end{aligned}
\end{equation}
for given \(a,b,c \in \mathbb{R}_+\), that
\begin{equation}
\tilde{\rho}:= \sum_{i=1}^{10} \ket{\psi_i} \bra{\psi_i}= E(a,b,c^{-1});
\end{equation}
which is the state obtained in Eq.(20) of Ref. \cite{relaxations} for \(p=0\). And hence by setting \(a=b=c^{-1}=\lambda\) we can write 
our states \(\rho(\lambda)\) as
\begin{equation}
\rho(\lambda)= \frac{\tilde{\rho}}{\Tr(\tilde{\rho})}.
\end{equation}
Finally note that our criteria in the main text also detects \(\rho(\lambda_1,\lambda_2,\lambda_3)\) for other instances than \(\lambda=\lambda_1=\lambda_2=\lambda_3\) but for the sake of simplicity we only consider the \(\rho(\lambda)= \rho(\lambda,\lambda,\lambda)\) cases in the main text.

\end{document}